\documentclass[11pt, reqno]{amsart}
	\usepackage[margin=1.3in]{geometry}

	\usepackage[utf8x]{inputenc}
	\usepackage{blindtext}
	\usepackage{graphicx}
	\usepackage{amsmath, amsthm, amssymb}
	\usepackage[english]{babel}
	\usepackage{marvosym}
	\usepackage{graphics}
	\usepackage{natbib}
	\usepackage{tikz}
	\usepackage{pgfplots}
	\usepgfplotslibrary{fillbetween}
	\usetikzlibrary{decorations.markings}
	\usepackage[colorlinks=true,
            linkcolor=blue,
            citecolor=blue,
            urlcolor=blue]{hyperref}
	\usepackage{mathtools}
	\usepackage{blindtext}
	\usepackage{comment}
	\usepackage{bm}
	\usepackage{subcaption}
	\newcommand\restr[2]{{% we make the whole thing an ordinary symbol
	  \left.\kern-\nulldelimiterspace % automatically resize the bar with \right
	  #1 % the function
	  \vphantom{\big|} % pretend it's a little taller at normal size
	  \right|_{#2} % this is the delimiter
	  }}

	\DeclareMathOperator{\divv}{div} 
	\DeclareMathOperator{\Poi}{Poi}

	\newtheorem{theorem}{Theorem}
	\newtheorem{corollary}{Corollary}
	\newtheorem{lemma}{Lemma}
	\newtheorem{definition}{Definition}

	\theoremstyle{definition}

	\newtheorem{remark}{Remark}
	
	\newtheorem*{example*}{Example}

\begin{document}

 % HEADER
 \title[Phase transitions in generalized XY models]{Phase transitions in generalized XY models}

 \makeatletter
 \@namedef{subjclassname@2020}{\textup{2020} Mathematics Subject Classification}
 \makeatother

 
 \author{Fabio Plaga}
 \address{Technische Universit\"at Wien}
 \email{fabio.plaga@tuwien.ac.at}
 
 
\begin{abstract}
	We associate height function models to a range of generalized two-dimensional XY models and prove that delocalisation of the height function model rules out exponential decay of the nematic order parameter in the primal spin model.
	The argument is based on a generalized loop representation which relates the variance of a height difference between two faces to the nematic order parameter.  
	This generalizes the results by van Engelenburg and Lis \cite{lis:2023}. 
	The range of models also includes the much-studied generalized XY model as proposed by Korshunov and independently by Lee and Grinstein, where neighbouring spins interact ferromagnetically favouring parallel alignment and nematically favouring parallel or antiparallel alignment. 
	Moreover, we recover the existence of a nematic region by a classical comparison argument of Pfister \cite{pfister:1982}.
\end{abstract}

\maketitle

\setcounter{tocdepth}{1}

	\section{Introduction}
	Generalizations of the classical XY model were first introduced in the 1980s by Korshunov and independently by Lee and Grinstein among others in \cite{korshunov:1986}, \cite{lee:1985}, \cite{sluckin:1988} and \cite{carpenter:1989}.
	In these models, classical two-component spins interact both ferromagnetically and nematically.
	In its simplest form, the Hamiltonian of such a model is parametrised by 
	\begin{equation}
	\label{eq:H_Delta}
	H_\Delta
	= 
	- 
	\sum_{u \sim v} 
	\left[
	\Delta
	\cos (\theta_u - \theta_v) 
	+ 
	(1 - \Delta ) 
	\cos (2 (\theta_u - \theta_v) )
	\right]
	, 
	\ 
	0 \leq \Delta \leq 1,
	\end{equation}
	where $\theta_v \in [0, 2 \pi)$ describes the spin orientation at vertex $v$ and the sum ranges over all neighbouring vertices of a finite graph $G = (V,E)$.
	We write $u \sim v$ if $uv \in E$. 
	The $2 \pi$-periodic term represents the \textit{ferromagnetic} coupling and favours parallel alignment of neighbouring spins, whereas the $\pi$-periodic term is the \textit{nematic} coupling and favours parallel and antiparallel alignment of neighbouring spins.
	\par 
	From a physical perspective, the generalized XY model amounts to a basic example of studying vortex excitations and their unbinding and therefore possible phase transitions at finite temperature, where, in addition to the integer vortices of the ferromagnetic coupling, fractional vortex excitations are present due to the nematic coupling. 
	In our example \eqref{eq:H_Delta} the $\pi$-periodic coupling term gives rise to half-vortices, where neighbouring half-vortices are connected by domain-wall strings with finite string tension.
	\par 
	Generalized XY models with nematic interaction have been used in the study of high-temperature superconductivity \cite{hlubina:2007}, arrays of unconventional Josephson junctions \cite{korshunov:2010} and liquid crystal films \cite{lee:1985} to name a few.
	\par 
	The generalization enriches the phase diagram compared to the standard XY model.
	Let $\beta \in [0, \infty]$ be the inverse temperature of the system.
	In order to characterise the phases in the $\Delta-\beta$ plane, one considers the following two correlation functions,
	\begin{align}
	\label{eq:G_1} 
	G_1(v) 
	& = 
	\langle \cos (\theta_o - \theta_v) \rangle^{H_\Delta}_{\Gamma, \beta} 
	\coloneqq 
	\lim_{G \nearrow \Gamma} 
	\langle \cos (\theta_o - \theta_v) \rangle^{H_\Delta}_{G, \beta} , \ \text{and}
	\\ 
	\label{eq:G_2}
	G_2(v)
	& =
	\langle \cos ( 2 (\theta_o - \theta_v)) \rangle^{H_\Delta}_{\Gamma, \beta} 
	\coloneqq 
	\lim_{G \nearrow \Gamma} 
	\langle \cos ( 2 (\theta_o - \theta_v)) \rangle^{H_\Delta}_{G, \beta},
	\end{align} 
	where $\Gamma$ is a
	planar and locally finite graph that is invariant under the action of a lattice isomorphic to $\mathbb{Z}^2$,
	$o$ is a fixed vertex and 
	$\langle \cdot \rangle_{G, \beta}^{H_\Delta}$ denotes the expectation under $H_\Delta$ at inverse temperature $\beta$ on $G$, see \eqref{eq:gibbs meas} for a precise definition.
	The sequences on the right-hand side of 
	\eqref{eq:G_1} and \eqref{eq:G_2} are non-decreasing by Ginibre inequalities and therefore the limits exist and are well-defined, see Remark \ref{re:gin}.
	\par 
	$G_1$ is sensitive to the ferromagnetic order, whereas $G_2$ is sensitive to both ferromagnetic and nematic order. 
	$G_1$ (respectively $G_2$) will be referred to as the \textit{ferromagnetic} (respectively \textit{nematic}) \textit{order parameter}.
	\par 
	 	Heuristically, when $\Delta \approx 1$ the model is a perturbation of the standard XY model, where the well-studied \textit{Berezinskii--Kosterlitz--Thouless} (BKT) phase transition \cite{kosterlitz:1973} occurs.
	Therefore one expects a \textit{disordered} regime at high temperatures, in which
	$G_1$, respectively 
	$G_2$, decay exponentially in the distance,
	meaning
	\[ 
	G_i(v) \leq C e^{- c |v|}, 
	\ 
	\text{for some 
	$c = c(\Gamma,\beta, \Delta, i),
	C = C(\Gamma,\beta, \Delta, i) >0$},
	\] 
	for $i = 1,2$
	and a \textit{BKT} regime at low temperatures, in which $G_1$, respectively $G_2$,
	decay \textit{slowly}
	at an inverse temperature-dependent power-law, 
	that is
	\[ 
	G_i(v) 
	\sim 
	|v|^{-\eta_i(\beta)}, 
	\]
	for $i = 1,2$.
	Here, $|v|$ denotes the geodesic graph distance of $v$ to the fixed vertex $o$.
	In particular, in the BKT phase $G_1$ and $G_2$ do not decay exponentially.
	\par  
	When $\Delta \approx 0$ one expects a third \textit{nematic} phase between the disordered and BKT phases. 
	It is characterized by exponential decay of $G_1$ and absence of exponential decay of $G_2$. 
	\\
	This behaviour is expected because on the nematic line corresponding to $\Delta = 0$,  
	\[ 
	\langle \cos ( \theta_u - \theta_v ) 
	\rangle^{H_0}_{\Gamma, \beta}
	= 
	0 ,
	\]
	for $u \neq v´$ by symmetry,
	whereas 
	\[ 
	\langle 
	\cos (2 ( \theta_u - \theta_v) ) 
	\rangle^{H_0}_{\Gamma, \beta}
	= 
	\langle 
	\cos ( \theta_u - \theta_v) 
	\rangle^{H_1}_{\Gamma, \beta},
	\]
	by a change of variables. 
	\\ 
	This prediction was observed numerically in \cite{carpenter:1989}, \cite{hubscher:2013} and \cite{sluckin:1988} using Monte Carlo methods.
	See Figure \ref{fig:schematic fig} for a schematic view of the predicted enriched $\Delta - \beta$ phase diagram.

	\begin{figure}[t] 
	\begin{tikzpicture}
	  \begin{axis}[
	    xlabel= $\Delta$,
	    ylabel= $\beta$,
	    ylabel style={rotate = 270},
	    axis lines* = left,
	    width=10cm,
	    height=8cm,
	    xmin = 0, 
	    xmax = 1, 
	    ymin = 0, 
	    ymax = 5,
	    x dir=reverse,
	   ytick={0,1,...,5},        
	minor y tick num=4,         % nine minor ticks → spacing 0.01
	ymajorgrids=true,
	  ]
	    % two phase boundaries
	    \addplot[name path=curve1, smooth, thick] coordinates {
	        (1, 1.1136)
	        (0.3, 1.5385)
	        (0, 10)};
	    \addplot[name path=curve2, smooth, thick] coordinates {
	        (0.3, 1.5385)
	        (0.1, 1.2903)
	        (0, 1.1136)};
	     \addplot[name path=top, smooth, forget plot] coordinates {(1,5) (0,5)};
	    \addplot[name path=bottom, smooth, forget plot] coordinates {(1,0) (0,0)};
	    % region labels	
	    \addplot [white!30, opacity =0.7] fill between[of=curve1 and top];
	     \addplot [red!30, opacity = 0.7] fill between[of=curve1 and curve2, soft clip={domain=0.3:0}];
	    \addplot [blue!30, opacity = 0.7] fill between[of=curve1 and bottom, soft clip={domain=1:0.3}];  
	    \addplot [blue!30, opacity = 0.7] fill between[of=curve2 and bottom, soft clip={domain=0.3001:0}];

	    \node at (axis cs:0.6, 3.5) {BKT};
	    \node at (axis cs:0.25, 0.5) {disordered};
	    \node at (axis cs:0.1, 2.5) {nematic};
	  \end{axis}
	\end{tikzpicture}
	\caption{Schematic $\Delta-\beta$ phase diagram adapted from simulations of \cite{carpenter:1989}, \cite{hubscher:2013}, \cite{sluckin:1988}} 
	\label{fig:schematic fig}
	\end{figure}
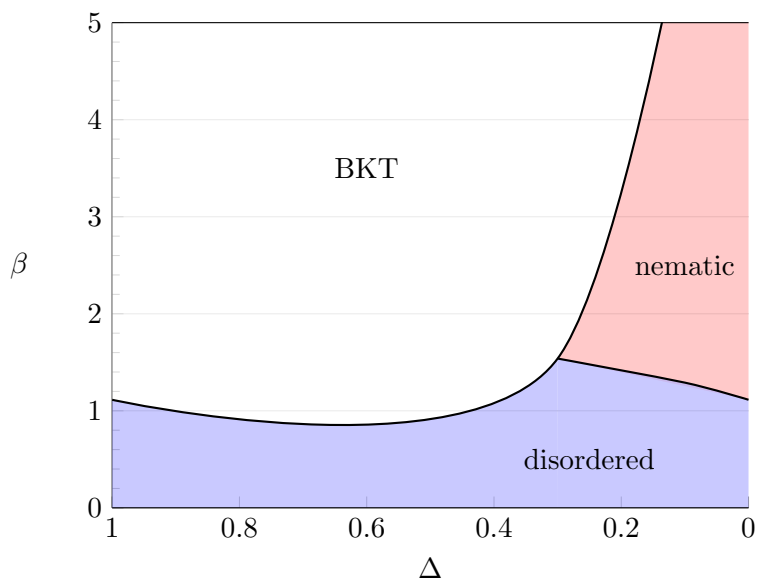

	One may also consider further generalizations of the XY model. 
	A natural class of Hamiltonians is given by the finite non-negative Fourier sums,
	\begin{equation}
	\label{eq:H Fourier}
	H = - 
	\sum_{u \sim v} 
	\sum_{k = 1}^q 
	J_k 
	\cos ( k (\theta_u - \theta_v) )
	\ 
	, q \in \mathbb{N}
	, 
	\ J_k \geq 0.
	\end{equation}
	Each coupling term corresponding to a positive coupling constant $J_k$ favours alignment of neighbouring spins up to modulo $2 \pi / k$.
	$G_1$ and $G_2$ are defined analogously to \eqref{eq:G_1} and \eqref{eq:G_2}.
	\\
	Some of these are investigated in the physics literature, see for example \cite{canova:2016}, \cite{poderoso:2011}, \cite{vzukovivc:2018multiple} and \cite{vzukovivc:2018}.
	This leads to even richer phase diagrams.	
	\par 
	The standard XY model has a dual model of integer-valued height functions 
	$\varphi : V^\dagger \rightarrow \mathbb{Z}$ defined on the respective dual graph 
	$G^\dagger = (V^\dagger, E^\dagger)$.
	\par 
	It is obtained by introducing the \textit{random current expansion} of the standard XY model and observing  a one-to-one correspondence between sourceless 
	integer 
	currents 
	and divergence-free 
	antisymmetric 
	flows. 
	Such flows are in bijection with integer-valued \textit{height functions} defined on the faces of $G$ and set to zero on the outer face.
	Summing over currents with a prescribed net flow induces the dual model of integer-valued height functions.
	\par 
	The height function model in infinite volume can exhibit two phases. 
	Depending on the inverse temperature $\beta \in [0, \infty]$, the variance of the height difference 
	$\varphi(x) - \varphi(y)$ between two faces $x,y \in V^\dagger$ remains either uniformly bounded, which we refer to as the \textit{localised} phase, or is unbounded as $|x-y| \rightarrow \infty$, the \textit{delocalised} phase.
	For precise definitions see Section \ref{sec:current exp}. 
	\par 
	In the standard XY model van Engelenburg and Lis \cite{lis:2023} used a loop representation to relate the variance of the height difference between two faces in the dual height function model with the correlation function $G_1$ in the primal spin model.
	Using this loop representation they deduced that delocalisation of the height function model rules out exponential decay of $G_1$.
	Moreover they showed that the height function model delocalises at finite inverse temperature $\beta_0$ by checking a sufficient condition given by Lammers in \cite{lammers:2022height}.
	That is, there exists some $\beta_0 \in (0, \infty)$ such that the height function model is delocalised for all $\beta \geq \beta_0$.
	To prove the existence of a phase of exponential decay is standard.
	Combined with the previous implication this yields the existence of a phase transition between exponential and non-exponential decay of $G_1$ in the standard XY model at finite inverse temperature. 
	\par 
	Aizenman et al. obtained the same result in \cite{aizenman:2021} by relating integer-restricted Gaussian fields to the BKT phase of two-component spin models, including the standard XY model.
	\subsection*{Main results}
	The main purpose of the present paper is to extend this result beyond the standard XY interaction.
	 We introduce an enhanced random current expansion, extend the height function formalism and define a generalized loop representation adapted to Hamiltonians with nematic coupling of the form 
	\begin{equation}
	\label{eq:H tilde}
	\tilde{H} = 
	 - 
	\sum_{u \sim v} 
	\sum_{k = 1}^q 
	\tilde{J}_k
	\cos^k ( ( \theta_u - \theta_v) ), 
	\ 
	\text{with $ q\in \mathbb{N}$ and $\tilde{J}_k \geq 0$}.
	\end{equation}
	At the level of Gibbs measures the generalized XY model with a Hamiltonian of type \eqref{eq:H tilde} is a subclass of generalized XY models with Hamiltonians of type 
	\eqref{eq:H Fourier} but crucially includes the Hamiltonians $H_\Delta$, $0 \leq \Delta \leq 1$ because of the identity 
	$\cos (2 \theta) = 2 \cos (\theta)^2 - 1$. 
	See Section \ref{sec:definition} for further  details.
	\par 
	We prove an identity relating the variance of the height difference between two faces to the nematic order parameter $G_2$ for Hamiltonians of type \eqref{eq:H tilde} and combine this with a delocalisation argument of the height function model to deduce the following generalized result: 
	
	\begin{theorem} 
	\label{th:main theorem 1}
	Consider the generalized XY model on $\Gamma$ with a Hamiltonian of the form \eqref{eq:H tilde} with some $\tilde{J}_k > 0$. 
	At every inverse temperature $\beta$ for which the associated dual height function model is delocalised in the sense of \eqref{eq:deloc var}, $G_2$ does not decay exponentially. 
	Moreover, there exists $\beta_0 < \infty$ such that height function model is delocalized for all $\beta \geq \beta_0$. 
	\end{theorem}
	
	It is standard to show the existence of a regime 
	$\beta \in [0, \beta_e)$, $\beta_e >0$, in which $G_2$ decays exponentially. 
	One can for example adapt the proof of Section 2.4 of \cite{peled:2019}. 
	In particular, this shows that $G_2$ undergoes a phase transition between a high temperature phase of exponential decay, and a low-temperature phase of no exponential decay. 
	
	\begin{remark}
	\label{re:stand XY} 
	In the standard XY model on $\Gamma$, absence of exponential decay of $G_1$ implies a uniform power-law lower bound on $G_1$ of the form
	\begin{equation}
	\label{eq:power law}
	G_1(v) \geq \frac{c(\Gamma)}{|v|},
	\end{equation}
	for $v \neq 0$ and some $c(\Gamma) > 0$.
	For example, $c(\mathbb{Z}^2) = \frac{1}{8}$.
	This is proven in \cite{lis:2023} using the Messager--Miracle--Sole (MMS) inequality \cite{messager:1977}. 
	\par 
	Also, in the standard XY model one can show that the phase transition is \textit{sharp}, in the sense that there is no other type of decay than exponential and power-law decay. 
	The proof uses the \textit{Simon--Lieb inequality} \cite{simon:1980} \cite{lieb:1980}.	 
	This way it was shown in \cite{lis:2023} that indeed a sharp BKT phase transition occurs in the standard XY model. 
	\par
	However, we note that both the MMS inequality and the Simon--Lieb inequality are not available for all Hamiltonians of type \eqref{eq:H tilde}.
	\end{remark} 
	
	\begin{remark} 
	In a later work by van Engelenburg and Lis \cite{engelenburg:2025}, the spin-height duality for positive definite integer valued-valued height potentials and $\mathcal{O}(2)$-symmetric spin models has been developed in considerable generality. 
	In particular, they obtained the implication that height function delocalisation implies non-exponential decay of an explicit, but non-standard, correlation function. 
	The advantage of our result is that it transparently yields non-exponential decay of $G_2$, the relevant nematic order parameter.
	\end{remark} 
	
	\begin{remark} 
	We remark that in the generalized XY model on $\Gamma$ with Hamiltonian $\tilde{H} \in \tilde{\mathcal{H}}$,
	$G_1$ and $G_2$ decay at least algebraically in the distance, i.e. for every $\beta > 0$ there exists 
	$c = c(\Gamma, \beta), \mu =  \mu(\Gamma, \beta)>0$ such that 
	\[ 
	| 
	G_i(v)
	| 
	\leq 
	c |v|^{-\mu}, 
	\ 
	\text{for} 
	\ 
	i = 1,2.
	\]
	The statement for $G_1$ can be found in \cite{gagnebin:2014} and is an adaptation of the McBryan-Spencer approach for the XY model \cite{mcbryan:1977}. 
	A straightforward adaptation of the proof in \cite{gagnebin:2014} yields the statement for $G_2$.
	\end{remark}	
	
	For the Hamiltonians 
	$H_\Delta$, $0 \leq \Delta \leq 1$, we recover the following statement.
	
	\begin{corollary} 
	\label{cor:corollary}
	Consider the generalized XY model on $\Gamma$ with Hamiltonian 
	$H_\Delta$, $0 \leq \Delta \leq 1$.
	Delocalisation of the associated dual height function model at inverse temperature $\beta$ rules out exponential decay of $G_2$ at the same $\beta$.
	Moreover, there exists $\beta_0 < \infty$ such that height function model is delocalized for all $\beta \geq \beta_0$. 
	\end{corollary}	
	
	\begin{remark} 
	\label{re:ph tr}
	Let $\Gamma = \mathbb{Z}^2$. 
	Corollary \ref{cor:corollary} gives the existence of a low-temperature regime of no exponential decay of $G_2$ for all $0 \leq \Delta \leq 1$. 
	\par 
	However, for $0 < \Delta < 1$, one can directly compare to the standard XY model via Ginibre inequalities to strengthen the result.
	Together with Remark \ref{re:stand XY} it is straightforward to deduce the existence of a low-temperature BKT regime for $G_i$,
	$i = 1,2$, where 
	\[ 
	G_i(v) \geq \frac{1}{8|v|}, 
	\] 
	for $v \neq 0$.
	\par 
	Existence of a high-temperature exponential decay regime for $G_1$, $G_2$ and for $0 \leq \Delta \leq 1 $ is standard as before.
	\par 
	For $\Delta = 0$, $G_1(v) = 0$ for all $v \neq o$ by symmetry and $G_2$ is the ferromagnetic order parameter in the standard XY model up to a change of variables as noted above.
	\par 
	All of the above is in accordance with the predicted phase diagram.
	\end{remark} 
	
	We stress that the novelty of the first contribution lies in establishing that delocalisation of the dual height function model implies non-exponential decay in the primal spin model.
	Proving the existence of a phase transition in the primal spin model can be achieved by elementary means as pointed out in Remark \ref{re:ph tr}.
	\par 
 	As a somewhat independent complementary result we establish a nematic regime in a region consistent with the expected phase diagram of Figure \ref{fig:schematic fig} for the specific Hamiltonians $H_\Delta$, $0 \leq \Delta \leq 1$.
	The existence of such a nematic regime is already contained in Pfister's correlation comparison argument, see Remark 5.5(2) of \cite{pfister:1982}, after setting 
	$\mu_1 = \Delta, \mu_2 = 1 - \Delta$ and $\mu_4 = 0$ in his rotator model specified in (5.24) of \cite{pfister:1982}.
	We nevertheless record and prove the corresponding result for Hamiltonians
	$H_\Delta$, $0 \leq \Delta \leq 1$,
	for completeness. 
	\par
	Let 
	$\beta^{XY}_c \in (0, \infty)$ be the critical inverse temperature of the standard XY model on $\mathbb{Z}^2$ and let $o$ be the origin vertex of $\mathbb{Z}^2$.
	
	\begin{theorem} 
	\label{th:main theorem 2}
	Consider the generalized XY model on $\mathbb{Z}^2$ with Hamiltonian $H_\Delta$, 
	$0 \leq \Delta \leq 1$.
	For $\beta > \beta^{XY}_c$ there exists 
	$\Delta_0(\beta) \in (0,1)$ such that for all 
	$0 \leq \Delta \leq \Delta_0(\beta)$,
	$G_1$ decays exponentially while 
	\[ 
	G_2(v) \geq \frac{1}{8 |v|},
	\]
	for $v \neq o$.
	\end{theorem}
	
	The proof uses an upper bound (Theorem \ref{th:comp AT}) on the correlation function $G_1$ by a two-point correlation of the \textit{Ashkin--Teller model} on $\mathbb{Z}^2$, whose phase diagram is known rigorously \cite{aoun:2024}, to deduce exponential decay of $G_1$ as well as a power-law lower bound on $G_2$ obtained by comparison with the standard XY model using Ginibre inequalities in the stated regime.
	Note that in Theorem \ref{th:main theorem 2} we have the stronger power-law lower bound on $G_2$, not just absence of exponential decay. 
	This amounts to a quantitative improvement of Pfister's result in \cite{pfister:1982}.
	\subsection*{Organization of the paper}
	Section \ref{sec:definition} gives the basic definitions and recalls the Ginibre inequalities and the Ashkin--Teller model. 
	Section \ref{sec:random current} introduces the random-current expansion and the dual height model. 
	Section \ref{sec:deloc of height} proves delocalisation of the height model at finite inverse temperature. 
	Section \ref{sec:loop repr} develops a loop representation, derives the implication from height delocalisation to absence of exponential decay and records a new correlation inequality (Theorem \ref{th:corr ineq}).
	Section \ref{sec:main results} proves the main theorems.

\section*{Acknowledgements}
	The author is grateful to Marcin Lis for valuable discussions, insightful suggestions, and careful supervision during the course of this work. 
	The author also thanks Fabio Toninelli, Kieran Ryan, Lorca Heeney and Diederik van Engelenburg for helpful comments on an earlier version of the manuscript.

	\section{Preliminaries} 
	\label{sec:definition}
	Let $G = (V,E)$ be a finite graph.
	The \textit{generalized XY model} is defined by the Gibbs distribution on $[0, 2 \pi)^{V}$,
	\begin{equation} 
	\label{eq:gibbs meas}
	d \mu^H_{G, \beta} (\theta) 
	\coloneqq 
	\frac{1}{Z_{G, \beta}} 
	e^{- \beta H } 
	\prod_{v \in V} d \theta_v.
	\end{equation}
	The $d \theta_v$, $v \in V$, are independent and uniform on $[0,2 \pi)$ \textit{a priori distributions} and $\beta \in [0, \infty]$ is the \textit{inverse temperature}. 	
	The Hamiltonian $H$ is taken from a class of Hamiltonians 
	\begin{equation}
	\label{eq:hamil}
	\mathcal{H} \coloneqq 
	\{ 
	- \sum_{u \sim v} 
	\sum_{k = 1}^q J_k \cos(k(\theta_u - \theta_v))
	\vert q \in \mathbb{N}, J_k \geq 0, k = 1 , \cdots, q 
	\},
	\end{equation}
	where the non-negative $J_k, k = 1, \cdots, q$, are called \textit{coupling constants}. 
	We write $u \sim v$ if $uv \in E$.
	One can also make the coupling constants edge dependent.
	As long as they are uniformly bounded, this does not create any conceptual differences.  
	Finally,
	\begin{equation}
	\label{eq:partition func}
	Z^H_{G,\beta} \coloneqq \int e^{- \beta H} \prod_{v \in V} d \theta_v
	\end{equation}
	is the \textit{partition function}.
	We will denote the expectation of an observable with respect to 
	$d \mu^H_{G,\beta}$ by 
	$\langle \cdot \rangle^H_{G, \beta}$ and the unnormalized expectation by 
	$\left( \cdot \right)^H_{G, \beta}
	\coloneqq 
	Z^H_{G, \beta} \langle \cdot \rangle^H_{G, \beta}$.
	\par 
	Clearly $(H_\Delta)_{0 \leq \Delta \leq 1} \subset \mathcal{H}$ and 
	$H_1$ is the Hamiltonian of the \textit{standard} XY model.
	\par 
	A function on $[0, 2 \pi)$ is called \textit{negative definite} if its Fourier coefficients are negative.
	
	\begin{lemma}[Ginibre inequalities \cite{ginibre:1970}]
	\label{le:ginibre}
	For any negative definite Hamiltonian $H$, 
	\begin{align} 
	\tag{First Ginibre inequality}
	\langle \cos (a \cdot \theta) \rangle^H_{G, \beta} \rangle & \geq 0 
	\\
	\tag{Second Ginibre inequality}
	\langle \cos (a \cdot \theta) \cos (b \cdot \theta) \rangle^H_{G, \beta} 
	& \geq 
	\langle \cos (a \cdot \theta) \rangle^H_{G, \beta} 
	\langle \cos (b \cdot \theta) \rangle^H_{G, \beta} .
	\end{align} 
	In this context, $a,b \in \mathbb{Z}^V$ and $a \cdot \theta = \sum_{v \in V} a_v \theta_v$. 
	\par 
	\end{lemma}
	
	\begin{remark}
	$\mathcal{H}$ is precisely the set of all real negative definite Hamiltonians.
	In particular, for all $H \in \mathcal{H}$ the first and second Ginibre inequality hold.
	\end{remark}
	
	\begin{remark} 
	\label{re:gin}
	$\langle \cos (\theta_o - \theta_v) \rangle^H_{G, \beta}$ and 
	$\langle \cos (2 (\theta_o - \theta_v)) \rangle^H_{G, \beta}$  
	are seen to be non-decreasing for increasing $G$ by differentiating
	with respect to an edge dependent coupling constant $J_{uw}$ 
	and using Lemma \ref{le:ginibre}.
	Therefore the correlation functions $G_1$ and $G_2$ as defined in \eqref{eq:G_1} and \eqref{eq:G_2} are well-defined and independent of the \textit{exhaustion} $G \nearrow \Gamma$.
	\end{remark}
	
	We will now introduce a new class of Hamiltonians similar to the one of the generalized XY model. 
	
	\begin{definition} 
	Define, 
	\begin{equation} 
	\tilde{\mathcal{H}} 
	\coloneqq 
	\{
	-
	\sum_{u \sim v} 
	\sum_{k = 1}^q
	\tilde{J}_k
	\cos^k (\theta_u - \theta_v) 
	\vert q \in \mathbb{N} , 
	\tilde{J}_k \geq 0, k = 1, \cdots , q 
	\}.
	\end{equation} 
	Each $\tilde{H} \in \tilde{\mathcal{H}}$ defines a probability distribution $d \mu^{\tilde{H}}_{G, \beta}$ on $[0, 2\pi)^V$  and partition function $Z^{\tilde{H}}_{G, \beta}$ analogously to \eqref{eq:gibbs meas} and \eqref{eq:partition func}.
	\end{definition}
	
	The new class of Hamiltonians $\tilde{\mathcal{H}}$ 
	is contained in $\mathcal{H}$. 
	To see this, note
	\begin{align*} 
	\cos^k(\theta_u - \theta_v) 
	& = 
	\left( 
	\frac{ e^{i(\theta_u - \theta_v)} + e^{-i(\theta_u - \theta_v)}}{2} \right)^k 
	\\
	&
	= 
	\frac{1}{2^k}
	\sum_{l = 0}^k 
	{k \choose l} 
	e^{i (2l - k) (\theta_u - \theta_v)}
	\\
	& 
	= 
	\frac{1}{2^k}
	\sum_{l = 0}^k 
	{k \choose l}
	\cos ((2l-k) (\theta_u - \theta_v)).
	\end{align*} 
	\par 
	The inclusion is strict as can be seen by expanding $\cos(n \theta)$ as a sum of powers of $\cos(\theta)$. 
	For example, $\cos (3 \theta) = 4 \cos^3 (\theta) - 3 \cos (\theta)$.
	\par
	Crucially, for a given 
	$H_\Delta$, $0 \leq \Delta \leq 1$, we can rewrite the nematic coupling term as
	\[ 
	\cos(2(\theta_u - \theta_v)) 
	= 2 \cos^2 (\theta_u - \theta_v) - 1. 
	\] 
	We can divide out the constant term as it appears in the partition function as well yielding 
	\[ 
	d \mu^{H_\Delta}_{G, \beta} 
	\sim 
	d \mu^{\tilde{H}_\Delta}_{G, \beta}, 
	\]
	where 
	\[ 
	\tilde{H}_{\Delta} = - \sum_{u \sim v} 
	\Delta \cos (\theta_u - \theta_v) 
	+ 
	2 (1 - \Delta) 
	\cos^2(\theta_u - \theta_v) 
	\in \tilde{\mathcal{H}}. 
	\]
	Therefore we can work with 
	$\tilde{H}_{\Delta} \in \tilde{\mathcal{H}}$ instead of $H_\Delta \in \mathcal{H}$. 
	\par 
	
	\begin{remark}
	The reason for introducing $\tilde{\mathcal{H}}$ is that it allows for a convenient random current expansion and loop representation, see Section \ref{sec:current exp} and Section \ref{sec:loop repr}. 
	The restriction ensures well-definedness when \textit{switching paths} (Definition \ref{def:path switching}) in a \textit{loop configuration} (Definition \ref{def:loop conf}) later on, see Remark \ref{re:just tilde H}.
	\end{remark}
	\subsection{Ashkin--Teller model}
	Given a finite graph $G = (V,E)$ the (isotropic) \textit{Ashkin--Teller} (AT) model is defined by the Gibbs distribution supported on pairs of spin configurations 
	$(s, t) \in \{ \pm 1 \}^V \times \{ \pm 1 \}^V$ ,
	\begin{equation} 
	d \mu^{AT}_\beta 
	( s, t ) 
	\coloneqq 
	\frac{1}{Z^{AT}_\beta} 
	e^{ 
	- \beta H(s, t)}, 
	\end{equation}
	where 
	\begin{equation} 
	H(s, t) 
	= 
	- 
	\sum_{u \sim v} 
	J
	\left( 
	s_u s_v + t_u t_v 
	\right)
	+ 
	U s_u s_v t_u t_v.
	\end{equation}
	$J$ and $U$ are the \textit{coupling constants} and are taken to be non-negative.
	$\beta \in [0, \infty]$ is again the \textit{inverse temperature} and the \textit{partition function}
	$Z^{AT}_\beta$ is the unique constant which ensures that 
	$\mu^{AT}_\beta$ is in fact a probability distribution. 
	We will denote the expectation of an observable with respect to 
	$d \mu^{AT}_{G, \beta}$ by 
	$\langle \cdot \rangle^{AT}_{G, \beta}$. 
	\par 	
	Aoun, Dober and Glazman deduced the full phase diagram of the isotropic AT model on $\mathbb{Z}^2$ in \cite{aoun:2024}.
	We record some useful results of this paper and refer the reader to \cite{aoun:2024} for further reading.
	\par
	Similar to before the infinite-volume limit is defined as
	\[ 
	\langle \cdot \rangle^{AT}_{\mathbb{Z}^2, \beta}
	\coloneqq 
	\lim_{G \nearrow \mathbb{Z}^2} 
	\langle \cdot \rangle^{AT}_{G, \beta}
	\] 
	for nice enough observables.
	$o$ is still the origin vertex.
	Well-definedness is ensured by the corresponding Ginibre inequality.
	\par 
	For fixed $J,U \geq 0$ there exist two critical $\beta_c^{(1)}, \beta_c^{(2)} > 0$ and $c, C > 0$ 
	such that 
	\begin{equation} 
	\langle 
	s_o s_v 
	\rangle^{AT}_{\mathbb{Z}^2, \beta}
	\begin{cases} 
	& \leq e^{- c |v|} \ \text{if} \ \beta < \beta_c^{(1)}
	\\ 
	& \geq C \ \text{if} \ \beta > \beta_c^{(1)}
	\end{cases}
	,
	\ \ \ \ \ 
	\langle 
	s_o t_o s_v t_v  
	\rangle^{AT}_{\mathbb{Z}^2, \beta}
	\begin{cases} 
	& \leq e^{- c |v|} \ \text{if} \ \beta < \beta_c^{(2)}
	\\ 
	& \geq C \ \text{if} \ \beta > \beta_c^{(2)}.
	\end{cases}
	\end{equation}
	Moreover,
	\[
	\beta_c^{(1)} \geq \beta_c^{(2)}. 
	\]

	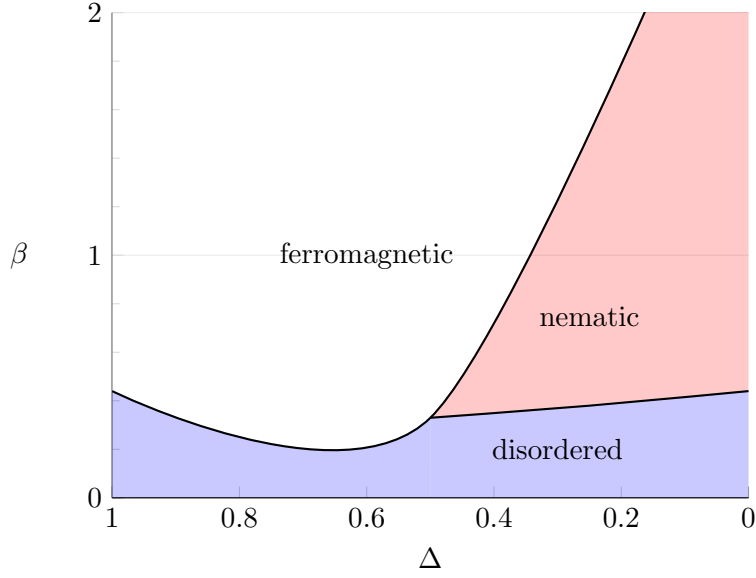
\begin{figure}[t] 
	\begin{tikzpicture}
	  \begin{axis}[
	    xlabel= $\Delta$,
	    ylabel= $\beta$,
	    ylabel style={rotate = 270},
	    axis lines* = left,
	    width=10cm,
	    height=8cm,
	    xmin = 0, 
	    xmax = 1, 
	    ymin = 0, 
	    ymax = 2,
	    x dir=reverse,
	   ytick={0,1,...,2},        
	minor y tick num=4,         % nine minor ticks → spacing 0.01
	ymajorgrids=true,
	  ]
	    % two phase boundaries
	    \addplot[name path=curve1, smooth, thick] coordinates {
	        (1, 0.44)
	        (0.5, 0.33)
	        (0, 3)};
	    \addplot[name path=curve2, smooth, thick] coordinates {
	        (0.5, 0.33)   
	        (0.25, 0.38)  
	        (0, 0.44)};
	     \addplot[name path=top, smooth, forget plot] coordinates {(1,5) (0,5)};
	    \addplot[name path=bottom, smooth, forget plot] coordinates {(1,0) (0,0)};
	    % region labels	
	    \addplot [white!30, opacity =0.7] fill between[of=curve1 and top];
	     \addplot [red!30, opacity = 0.7] fill between[of=curve1 and curve2, soft clip={domain=0.5:0}];
	    \addplot [blue!30, opacity = 0.7] fill between[of=curve1 and bottom, soft clip={domain=1:0.4999}];  
	    \addplot [blue!30, opacity = 0.7] fill between[of=curve2 and bottom, soft clip={domain=0.4999999:0}];

	    \node at (axis cs:0.6, 1) {ferromagnetic};
	    \node at (axis cs:0.25, 0.75) {nematic};
	    \node at (axis cs:0.3, 0.2) {disordered};
	  \end{axis}
	\end{tikzpicture}
	\caption{Schematic $\Delta-\beta$ phase diagram of the isotropic AT model on $\mathbb{Z}^2$ with $J = \Delta$ and $U = 1 - \Delta$} 
	\label{fig:schematic fig AT} 
	\end{figure}

	In Figure \ref{fig:schematic fig AT} 
	we plotted the schematic $\Delta - \beta$ phase diagram of the isotropic AT model on $\mathbb{Z}^2$ with $J = \Delta$ and $U = 1 - \Delta$ as proven in \cite{aoun:2024}.  
	In the disordered phase 
	$\langle s_o s_v \rangle_{\mathbb{Z}^2, \beta}^{AT}$ as well as 	
	$\langle s_o s_v t_o t_v \rangle_{\mathbb{Z}^2, \beta}^{AT}$
	decay exponentially for 
	$|v| \rightarrow \infty$, 
	while in the ferromagnetic phase they are both \textit{ordered}, i.e. bounded from below by a constant. 
	In the nematic phase 
	$\langle s_o s_v \rangle_{\mathbb{Z}^2, \beta}^{AT}$
	decays exponentially while 
	$\langle s_o s_v  t_o t_v \rangle_{\mathbb{Z}^2, \beta}^{AT}$ is ordered for
	$|v| \rightarrow \infty$.
	\par 
	Also note that  
	\begin{equation}
	\label{eq:crit temp}
	\beta_c^{(1)} \rightarrow \infty
	\ 
	\text{for the AT model on $\mathbb{Z}^2$ with 
	$J = \frac{\Delta}{2} ,
	U = (1- \Delta)$ 
	and 
	$\Delta \rightarrow 0$}.
	\end{equation}
	\par
	There is a striking resemblance between the expected phase diagram of the generalized XY model in Figure \ref{fig:schematic fig} and the established phase diagram of the isotropic AT model in Figure \ref{fig:schematic fig AT}.
	This resemblance is to be expected as the term in the Hamiltonian of the AT model proportional to $J$ represents a ferromagnetic coupling and the term proportional to $U$ represents a nematic coupling.
	Indeed, the AT model may be regarded as a discrete analogue of the generalized XY model with Hamiltonian $H_\Delta$.
	\par 
	In fact, $G_1$ (respectively $G_2$) is bounded above by the two-point (respectively four-point function) of the Ashkin--Teller model. 
	This will be useful when showing existence of the nematic phase in the generalized XY model.
	
	\begin{theorem} 
	\label{th:comp AT}
	Consider the generalized XY model with $H_\Delta$ and the Ashkin--Teller model with 
	$J = \frac{\Delta}{2}$ and $U = 1-\Delta$.
	For $u, v \in V$ it holds that 
	\begin{align} 
	\langle \cos (\theta_u - \theta_v) \rangle^{H_\Delta}_{G, \beta} 
	& \leq 
	\langle s_u s_v \rangle^{AT}_{G, \beta},
	\\
	\langle \cos ( 2 (\theta_u - \theta_v)) \rangle^{H_\Delta}_{G, \beta} 
	& \leq 
	\langle s_u s_v t_u t_v 
	\rangle^{AT}_{G, \beta}.
	\end{align}
	\begin{proof} 
	Identify 
	$\theta \in [0, 2 \pi) 
	\leftrightarrow 
	\sigma = (\sigma^{(1)}, \sigma^{(2)}) \in \mathbb{S}$, where $\mathbb{S}$ is the real unit circle, via
	$(\sigma^{(1)}, \sigma^{(2)}) = 
	(\cos(\theta), \sin (\theta))$. 
	\par 
	Define 
	\[ 
	H_\Delta^{\lambda} 
	\coloneqq 
	H_\Delta - \lambda 
	\sum_{u \in V} 
	\cos 
	\left( 
	4 \theta_u
	\right)
	.
	\]
	By Lemma \ref{le:ginibre},
	\begin{align}
	\label{eq:AT ineq 1} 
	\langle \cos (\theta_u - \theta_v) \rangle^{H_\Delta^0}_{G,\beta} 
	& \leq 
	\langle \cos (\theta_u - \theta_v) \rangle^{H_\Delta^\infty}_{G, \beta},
	\\ 
	\label{eq:AT ineq 2}
	\langle 
	\cos
	\left( 
	2  (\theta_u - \theta_v)
	\right) 
	\rangle^{H_\Delta^0}_{G,\beta} 
	& \leq 
	\langle
	\cos \left( 
	2 
	(\theta_u - \theta_v)
	\right) 
	\rangle^{H_\Delta^\infty}_{G, \beta}.
	\end{align} 
	For 
	$\lambda = 0$ we recover 
	$H_\Delta$ and taking 
	$\lambda \rightarrow \infty$ concentrates all mass on the angles 
	$\theta = 0, \frac{\pi}{4}, \frac{\pi}{2}, \frac{ 3 \pi}{4}$ or equivalently on
	$(\sigma^{(1)}, \sigma^{(2)}) = 
	(\pm 1, 0) , (0, \pm 1)$. 
	\par 
	Let 
	$s =
	\sigma^{(1)} + \sigma^{(2)}$ and 
	$t = 
	\sigma^{(1)} - \sigma^{(2)}$. 
	Then $\sigma^{(1)} = \frac{s + t}{2}$ and 
	$\sigma^{(2)} = \frac{s - t}{2}$. 
	This operation amounts to rotating and rescaling the support of the measure to 
	$(s,t) = (\pm 1, \pm 1)$.
	Write 
	\begin{align*} 
	\cos(\theta_u - \theta_v) 
	& = 
	\sigma^{(1)}_u \sigma^{(1)}_v
	+ 
	\sigma^{(2)}_u \sigma^{(2)}_v 
	\\ 
	& = 
	\frac{s_u s_v + t_u t_v}{2}, 
	\ \text{and},
	\\ 
	\cos(2(\theta_u - \theta_v)) 
	& = 
	2 \cos (\theta_u - \theta_v)^2 - 1
	\\
	& =
	2 
	\left( 
	\frac{s_u + t_u}{2} 
	\frac{s_v + t_v}{2} 
	+ 
	\frac{s_u - t_u}{2} 
	\frac{s_v - t_v}{2} 
	\right)^2
	- 1
	\\ 
	& = 
	\frac{1}{2} 
	\left( 
	s_u s_v + t_u t_v 
	\right)^2
	- 1
	\\ 
	& = 
	s_u s_v t_u t_v.
	\end{align*}
	Plugging this into \eqref{eq:AT ineq 1} and \eqref{eq:AT ineq 2} finishes the proof.
	\end{proof}
	\end{theorem}
		
	Note that if $\Delta = 1$ the latter theorem reduces to the previously known comparison between the standard XY model and the Ising model \cite{simon:1980comparison}.
	\section{Random current expansion and height function model}
	\label{sec:random current}
	\subsection{Random current expansion}
	\label{sec:current exp}
	In what follows we derive a random current
	expansion for 
	$\tilde{H} \in \tilde{\mathcal{H}}$. 
	In order to prevent notational overhead, we will restrict ourselves to the Hamiltonians of the form
	\begin{equation}
	\label{eq:H example}
	\tilde{H} = 
	- 
	\sum_{u \sim v} 
	J_1 
	\cos (\theta_u - \theta_v) 
	+ 
	J_2
	\cos^2 (\theta_u - \theta_v), 
	\ 
	\text{for} \ J_1, J_2 > 0.
	\end{equation}
	An analogous random current expansion holds for any $\tilde{H} \in \tilde{\mathcal{H}}$, see Remark \ref{rem:gen current}.
	\par 
	For a given graph $G=(V,E)$ let the set of directed edges be $\overrightarrow{E} = 
	\{ 
	(uv),(vu): uv \in E \}$.
	A \textit{current}  
	$\bm{n} : \overrightarrow{E} \rightarrow \mathbb{N}_0$ is an assignment from the set of directed edges to $\mathbb{N}_0$. 
	Given a current, we can define the \textit{amplitude}
	\[ 
	|\bm{n}|_{uv} \coloneqq \bm{n}_{(uv)} + \bm{n}_{(vu)}
	\] 
	and 
	the 
	\textit{divergence} 
	\[ 
	\delta \bm{n}_u 
	\coloneqq 
	\sum_{v \sim u} 
	\bm{n}_{(uv)} - \bm{n}_{(vu)}.
	\]
	Upon identifying $\theta \in [0, 2\pi) \leftrightarrow e^{i \sigma} \in \mathbb{S}$, the complex unit circle, in the usual fashion we can write
	\[
	\cos( \theta_u - \theta_v ) 
	=
	\frac{1}{2}
	\left(
	\sigma_u \bar{\sigma}_v
	+ 
	\bar{\sigma}_u \sigma_v 
	\right),
	\] 
	where $\sigma_u, \sigma_v \in \mathbb{S}$.
	Then
	\begin{align*}
	\cos^2(\theta_u - \theta_v) 
	& =
	\frac{1}{4}
	\left[
	(\sigma_u \bar{\sigma}_v)^2 
	+ 1 + 1 
	+ (\bar{\sigma}_u \sigma_v)^2
	\right]
	\end{align*} 
	and 
	\begin{align} 
	Z^{\tilde{H}}_{G, \beta} 
	& = 
	\int e^{- \beta \tilde{H}} \prod_{v \in V} d \sigma_v
	\\ 
	& = 
	\int \prod_{(uv) \in \overrightarrow{E}} 
	e^{\frac{\beta J_1}{2} \sigma_u \bar{\sigma}_v}
	e^{\frac{\beta J_2}{4} (\sigma_u \bar{\sigma}_v)^2} 
	e^{\frac{\beta J_2}{4}}
	\prod_{v \in V} d \sigma_v
	\\
	& =  	
	\sum_{
	\substack{
	\bm{n}, 
	\bm{m_1},
	\bm{m_2} 
	: 
	\overrightarrow{E} \rightarrow \mathbb{N}_0}
	}
	\prod_{(uv) \in \overrightarrow{E}} 
	\frac{( \beta J_1/2 )^{\bm{n}_{(uv)}}}{\bm{n}_{(uv)}!} 
	\frac{( \beta J_2/4 )^{{\bm{m_1}}_{(uv)}}}{{\bm{m_1}}_{(uv)}!} 
	\frac{( \beta J_2/4 )^{{\bm{m_2}}_{(uv)}}}{{\bm{m_2}}_{(uv)}!} 
	\prod_{v \in V} 
	\sigma_v^{\delta 
	\left( \bm{n} + 2 \bm{m_1} + 0 \bm{m_2}
	\right)_v
	} 
	d \sigma_v
	\label{eq:mit div1}
	\\ 
	& = 
	\sum_{ 
	\delta 
	\left( 
	\bm{n} + 2 \bm{m_1} + 0 \bm{m_2}
	\right) \equiv 0}
	\omega
	(\bm{n}, \bm{m_1}, \bm{m_2}) 
	\label{eq:mit div2}
	, 
	\end{align}
	where 
	\begin{equation} 
	\label{eq:weight current}
	w
	( \bm{n}, \bm{m_1}, \bm{m_2} )
	\coloneqq
	\prod_{uv \in E} 
	\frac{( \beta J_1/2 )^{|\bm{n}|_{uv}}}{\bm{n}_{(uv)}! \bm{n}_{(vu)}!} 
	\frac{( \beta J_2/4 )^{{|\bm{m_1}|}_{uv}}}{{\bm{m_1}}_{(uv)}!{\bm{m_1}}_{(vu)}!} 
	\frac{( \beta J_2/4 )^{{|\bm{m_2}|}_{uv}}}{{\bm{m_2}}_{(uv)}!{\bm{m_2}}_{(vu)}!},
	\end{equation}
	is the \textit{weight} of the tuple of currents $( \bm{n}, \bm{m_1}, \bm{m_2} )$.
	Let 
	\[
	\Omega_0(G) \coloneqq 
	\{
	(\bm{n},
	\bm{m_1}, \bm{m_2}) 
	\in 
	\left( 
	\mathbb{N}_0^{\overrightarrow{E}}
	\right)^3
	:
	\delta 
	(\bm{n} + 
	2 \bm{m_1} + 0 \bm{m_2}) \equiv 0
	\}
	\]
	be the set of \textit{sourceless} tuples of currents on $G$.
	
	\begin{definition}
	The weight in \eqref{eq:weight current} naturally 
	defines a probability measure $\mathbb{P}^{\tilde{H}}
	_{G,\beta}$ on $\Omega_0(G)$. 
	\end{definition}	
	
	Observe that in the expansion of the partition function \eqref{eq:mit div2} the contribution of $\bm{n}$, $\bm{m_1}$ and $\bm{m_2}$ in terms of the overall divergence is multiplied by a factor $1$, $2$ and $0$ respectively. 
	However the contribution of $\bm{n}$, $\bm{m_1}$ and $\bm{m_2}$ in the weight \eqref{eq:weight current} is equal for all three. 
	This prompts the definition of \textit{doubled} edges. 
	\par 
	A \textit{doubled edge} between $u$ and $v$ can be thought of as two undirected edges between $u$ and $v$ glued together, such that there is a distinct 'first' and 'second' edge.
	Directing the two edges result in four \textit{directed doubled edges} between $u$ and $v$:
	\begin{itemize}
	\item 
	the directed doubled edge, where both the first and second edge are directed from $u$ to $v$
	\item 
	its reversal
	\item 
	the directed doubled edge, where the first edge is directed from $u$ to $v$ and the second from $v$ to $u$
	\item 
	its reversal. 
	\end{itemize}

	\begin{figure}[t]
	\label{fig:new edges}
	\centering

	\begin{subfigure}[b]{0.3\textwidth}
	\centering

	\begin{tikzpicture}[scale =0.7,
	    every node/.style={circle, draw, inner sep=2pt},
	    edge/.style={line width=1.2pt},
	    midarrow/.style={
	        postaction={decorate},
	        decoration={
	            markings,
	            mark=at position 0.5 with {\arrow[red]{stealth}}
	        }
	    },
	    midarrowrev/.style={
	        postaction={decorate},
	        decoration={
	            markings,
	            mark=at position 0.5 with {\arrow[red]{stealth reversed}}
	        }
	    }
	]
	
	    % vertices
	    \node (u) at (0,0) {$u$};
	    \node (v) at (3,0) {$v$};
	
	    % offset from nodes
	    \def\off{0.18}
	
	    % shifted start/end points
	    \coordinate (uout) at ([xshift=\off cm]u);
	    \coordinate (vin)  at ([xshift=-\off cm]v);

	    % lower edge with red arrow ←
	    % IMPORTANT: no default '->', only decoration
	    \draw[edge, midarrow, blue]
	        (uout) .. controls +(1,0.5) and +(-1,0.5) .. (vin);
	
	\end{tikzpicture}
	\begin{tikzpicture}[scale =0.7,
	    every node/.style={circle, draw, inner sep=2pt},
	    edge/.style={line width=1.2pt},
	    midarrow/.style={
	        postaction={decorate},
	        decoration={
	            markings,
	            mark=at position 0.5 with {\arrow[red]{stealth}}
	        }
	    },
	    midarrowrev/.style={
	        postaction={decorate},
	        decoration={
	            markings,
	            mark=at position 0.5 with {\arrow[red]{stealth reversed}}
	        }
	    }
	]
	
	    % vertices
	    \node (u) at (0,0) {$u$};
	    \node (v) at (3,0) {$v$};
	
	    % offset from nodes
	    \def\off{0.18}
	
	    % shifted start/end points
	    \coordinate (uout) at ([xshift=\off cm]u);
	    \coordinate (vin)  at ([xshift=-\off cm]v);

	    % lower edge with red arrow ←
	    % IMPORTANT: no default '->', only decoration
	    \draw[edge, midarrowrev, blue]
	        (uout) .. controls +(1,0.5) and +(-1,0.5) .. (vin);
	
	\end{tikzpicture}
	
	\caption{Directed single edges. 
	\\ 
	Multiplicities are given 
	\\ 
	by $\bm{n}$ in our notation.}
	 
	\end{subfigure}
	\begin{subfigure}[b]{0.3\textwidth}
	\centering
	
	\begin{tikzpicture}[scale = 0.7,
	    every node/.style={circle, draw, inner sep=2pt},
	    edge/.style={line width=1.2pt},
	    midarrow/.style={
	        postaction={decorate},
	        decoration={
	            markings,
	            mark=at position 0.5 with {\arrow[red]{stealth}}
	        }
	    },
	    midarrowrev/.style={
	        postaction={decorate},
	        decoration={
	            markings,
	            mark=at position 0.5 with {\arrow[red]{stealth reversed}}
	        }
	    }
	]
	
	    % vertices
	    \node (u) at (0,0) {$u$};
	    \node (v) at (3,0) {$v$};
	
	    % offset from nodes
	    \def\off{0.18}
	
	    % shifted start/end points
	    \coordinate (uout) at ([xshift=\off cm]u);
	    \coordinate (vin)  at ([xshift=-\off cm]v);
	
	    % upper edge with red arrow →
	    \draw[edge, midarrow, blue]
	        (uout) .. controls +(1,1.2) and +(-1,1.2) .. (vin);
	
	    % lower edge with red arrow ←
	    % IMPORTANT: no default '->', only decoration
	    \draw[edge, midarrow, blue]
	        (uout) .. controls +(1,0.5) and +(-1,0.5) .. (vin);
	
	\end{tikzpicture}

	\begin{tikzpicture}[scale =0.7,
	    every node/.style={circle, draw, inner sep=2pt},
	    edge/.style={line width=1.2pt},
	    midarrow/.style={
	        postaction={decorate},
	        decoration={
	            markings,
	            mark=at position 0.5 with {\arrow[red]{stealth}}
	        }
	    },
	    midarrowrev/.style={
	        postaction={decorate},
	        decoration={
	            markings,
	            mark=at position 0.5 with {\arrow[red]{stealth reversed}}
	        }
	    }
	]
	
	    % vertices
	    \node (u) at (0,0) {$u$};
	    \node (v) at (3,0) {$v$};
	
	    % offset from nodes
	    \def\off{0.18}
	
	    % shifted start/end points
	    \coordinate (uout) at ([xshift=\off cm]u);
	    \coordinate (vin)  at ([xshift=-\off cm]v);
	
	    % upper edge with red arrow →
	    \draw[edge, midarrowrev, blue]
	        (uout) .. controls +(1,1.2) and +(-1,1.2) .. (vin);
	
	    % lower edge with red arrow ←
	    % IMPORTANT: no default '->', only decoration
	    \draw[edge, midarrowrev, blue]
	        (uout) .. controls +(1,0.5) and +(-1,0.5) .. (vin);
	
	\end{tikzpicture}
	
	\caption{Directed doubled edges. 
	\\ 
	Multiplicities are given 
	\\ 
	by $\bm{m_1}$ in our notation.}

	\end{subfigure}%
	\begin{subfigure}[b]{0.3\textwidth}
	\centering
	\begin{tikzpicture}[scale =0.7,
	    every node/.style={circle, draw, inner sep=2pt},
	    edge/.style={line width=1.2pt},
	    midarrow/.style={
	        postaction={decorate},
	        decoration={
	            markings,
	            mark=at position 0.5 with {\arrow[red]{stealth}}
	        }
	    },
	    midarrowrev/.style={
	        postaction={decorate},
	        decoration={
	            markings,
	            mark=at position 0.5 with {\arrow[red]{stealth reversed}}
	        }
	    }
	]
	
	    % vertices
	    \node (u) at (0,0) {$u$};
	    \node (v) at (3,0) {$v$};
	
	    % offset from nodes
	    \def\off{0.18}
	
	    % shifted start/end points
	    \coordinate (uout) at ([xshift=\off cm]u);
	    \coordinate (vin)  at ([xshift=-\off cm]v);
	
	    % upper edge with red arrow →
	    \draw[edge, midarrow, blue]
	        (uout) .. controls +(1,1.2) and +(-1,1.2) .. (vin);
	
	    % lower edge with red arrow ←
	    % IMPORTANT: no default '->', only decoration
	    \draw[edge, midarrowrev, blue]
	        (uout) .. controls +(1,0.5) and +(-1,0.5) .. (vin);
	
	\end{tikzpicture}
	
	\begin{tikzpicture}[scale =0.7,
	    every node/.style={circle, draw, inner sep=2pt},
	    edge/.style={line width=1.2pt},
	    midarrow/.style={
	        postaction={decorate},
	        decoration={
	            markings,
	            mark=at position 0.5 with {\arrow[red]{stealth}}
	        }
	    },
	    midarrowrev/.style={
	        postaction={decorate},
	        decoration={
	            markings,
	            mark=at position 0.5 with {\arrow[red]{stealth reversed}}
	        }
	    }
	]
	
	    % vertices
	    \node (u) at (0,0) {$u$};
	    \node (v) at (3,0) {$v$};
	
	    % offset from nodes
	    \def\off{0.18}
	
	    % shifted start/end points
	    \coordinate (uout) at ([xshift=\off cm]u);
	    \coordinate (vin)  at ([xshift=-\off cm]v);
	
	    % upper edge with red arrow →
	    \draw[edge, midarrowrev, blue]
	        (uout) .. controls +(1,1.2) and +(-1,1.2) .. (vin);
	
	    % lower edge with red arrow ←
	    % IMPORTANT: no default '->', only decoration
	    \draw[edge, midarrow, blue]
	        (uout) .. controls +(1,0.5) and +(-1,0.5) .. (vin);
	
	\end{tikzpicture}
	
	\caption{Directed doubled edges. 
	\\
	Multiplicities are given 
	\\ 
	by $\bm{m_2}$ in our notation.}
	\end{subfigure}
	\hfill 
	
	\caption{Directed edge types between $u$ and $v$ for Hamiltonians of type \eqref{eq:H example}.}
	\label{fig:edge types}
	\end{figure}
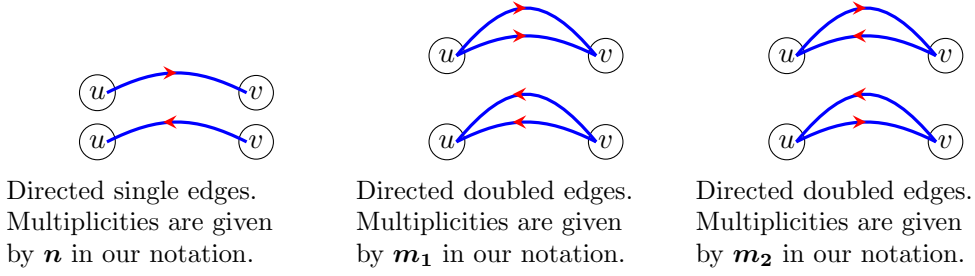

	Doubled edges give rise to a new type of multigraph:
	
	\begin{definition}
	An \textit{undirected multigraph on $G$ with doubled 
	edges $\mathcal{M}$} is a graph on $G$ with possibly 
	multiple undirected single and doubled edges between 
	$u$ and $v$ if $uv \in E$. 
	\end{definition}
	
	\begin{definition}
	A \textit{directed multigraph on $G$ with doubled 
	edges $\overrightarrow{\mathcal{M}}$} is a graph on 
	$G$ with possibly multiple directed single and 
	directed doubled edges between $u$ and $v$ if $uv \in 
	E$. 
	\end{definition}
	
	A directed multigraph with doubled edges \textit{on} 
	$\mathcal{M}$ is obtained from an undirected 
	multigraph with doubled edges $\mathcal{M}$ by 
	assigning orientations to the single and doubled 
	edges. 
	\par 	
	A tuple of currents $\bm{N} = 
	(\bm{n}, \bm{m_1}, \bm{m_2})$ \textit{induce} an undirected 
	multigraph with doubled edges 
	$\mathcal{M}^{\bm{N}}$ by 
	setting
	\begin{equation}
	\label{not:single}
	s(\mathcal{M}^{\bm{N}})_{uv} 
	\coloneqq |\bm{n}|_{uv},
	\end{equation}
	where $s(\mathcal{M}^{\bm{N}})_{uv}$ 
	is the number of undirected single edges between $u$ 
	and $v$ and
	\begin{equation}
	\label{not:doubled}
	d(\mathcal{M}^{\bm{N}})_{uv} 
	\coloneqq 
	|\bm{m_1}|_{uv}
	+
	|\bm{m_2}|_{uv}
	\end{equation}
	where $d(\mathcal{M}^{\bm{N}})_{uv}$ 
	is the number of undirected doubled edges between $u$ 
	and $v$.
	\par 
	We say that the directed multigraph with doubled 
	edges 
	$\overrightarrow{\mathcal{M}}^{\bm{N}}$ on 
	$\mathcal{M}^{\bm{N}}$ is \textit{induced} by 
	$
	\bm{N} = \left( 
	\bm{n}, \bm{m_1}, \bm{m_2} \right)$ 
	if for every $(uv) \in 
	\overrightarrow{E}$ there are exactly 
	\begin{itemize}
	\item $\bm{n}_{(uv)}$ directed single edges from 
	$u$ to $v$,
	\item ${\bm{m_1}}_{(uv)}$ directed doubled 
	edges, where the 'first' and the 'second' edge face 
	from $u$ to $v$
	\item ${\bm{m_2}}_{(uv)}$ directed doubled 
	edges, where the 'first' edge faces from $u$ to $v$ 
	and the 'second' edge faces from $v$ to $u$.
	\end{itemize}
	\par 
	Note that there are multiple 
	$\overrightarrow{\mathcal{M}}^{\bm{N}}$ on 
	$\mathcal{M}^{\bm{N}}$ and that they are equal up to reordering the directed single and directed doubled edges between two neighbouring vertices. 
	The precise number will be important for the loop representation of Section \ref{sec:loop repr}. 
	\par 
	For now one can think a tuple of currents
	$\bm{N}= ( \bm{n}, \bm{m_1}, \bm{m_2}) \in \Omega_0(G)$ as a directed multigraph with doubled edges 
	$\overrightarrow{\mathcal{M}}^{\bm{N}}$ 
	on 
	$\mathcal{M}^{\bm{N}}$.
	
	\begin{remark} 
	\label{rem:gen current}
	One can treat general Hamiltonians $\tilde{H} \in \tilde{\mathcal{H}}$ with at least one $J_k > 0$ as follows: 
	\par 
	For each $J_k > 0$ in $\tilde{H}$ define the space 
	$\{ \pm 1 \}^k$. 
	For a sequence $s_k \in \{\pm 1 \}^k$ define the equivalence relation 
	\[ 
	s_k \sim \tilde{s}_k 
	: \Leftrightarrow 
	s_k = - \tilde{s}_k, 
	\]
	where the "$-$" is understood pointwise.
	The equivalence classes of $\{ \pm 1 \}^k / \sim$ are denoted by $[ s_k ]$ and 
	\[ 
	[S]_{\tilde{H}} \coloneqq 
	\bigcup_{k:J_k > 0 \
	\text{in} 
	\ 
	\tilde{H}}
	\{\pm 1\}^k 
	/ \sim . 
	\]	 
	Define the \textit{absolute divergence} of $[ s_k ]$ by $\divv 
	[ s_k ] \coloneqq | \sum_{i = 1}^k (s_k)_i |$. 
	Well-definedness is checked immediately.
	\par
	Then we can expand each coupling term of the Hamiltonian,
	\begin{align*}
	J_k \cos^k ( \theta_u - \theta_v ) 
	& = 
	\frac{J_k}{2^k} 
	\sum_{s_k \in \{ \pm 1\}^k }
	\prod_{i = 1}^k 
	(\sigma_u \bar{\sigma}_v)^{(s_k)_i}
	\\
	& 
	= 
	\frac{J_k}{2^k} 
	\sum_{[s_k] \in \{ \pm 1\}^k / \sim }
	\left(
	\prod_{i = 1}^k 
	(\sigma_u \bar{\sigma}_v)^{(s_k)_i}
	+ 
	\prod_{i = 1}^k 
	(\bar{\sigma}_u \sigma_v)^{(s_k)_i}
	\right).
	\end{align*}
	Thus the partition function expands into a sum of weights 
	\[ 
	Z^{\tilde{H}}_{G, \beta} 
	= 
	\sum_{ 
	\delta 
	\left( 
	\sum_{[s] \in [S]_{\tilde{H}}}
	\divv [s] 
	\bm{n}_{[s]}	
	\right)
	\equiv 0}
	w 
	\left( 
	\left( 
	\bm{n}_{[s]}
	\right)_{[s] \in [S]_{\tilde{H}}}
	\right), 
	\]
	where 
	\[ 
	w 
	\left(
	\left( 
	\bm{n}_{[s]}
	\right)_{[s] \in [S]_{\tilde{H}}}
	\right)
	= 
	\prod_{[s] \in [S]_{\tilde{H}}} 
	\frac{
	(\beta J_{\ell([s])}/ 2^{\ell([s])})
	^{| \bm{n_{[s]}}  |_{uv}}}
	{	
	{\bm{n}_{[s]}}_{(uv)}! 
	{\bm{n}_{[s]}}_{(vu)}!
	} .
	\]
	$\ell ([s])$ is the length of the sequence of any representative of $[s]$. 
	\par 
	Graphically: 
	Each $J_k > 0$ gives rise to \textit{$k$-edges}.
	A \textit{$k$-edge} consists of $k$ labelled edges, say from $1$ to $k$, glued together. 
	Orienting the labelled edges creates $2^k$ \textit{oriented $k$-edges}.  
	Note that there are $2^{k-1}$ many directed $k$-edges up to reversal. 
	\\ 
	Directed and undirected multigraphs whose edges types are the $k$-edges for those $k$ with $J_k > 0$ are defined analogously to before.
	We stress that such multigraphs contain edges of different arities, with $k$-edges present only if $J_k > 0$.
	\\
	Instead of a tuple of currents $([s])_{[s] \in [S]_{\tilde{H}}}$ one may equivalently consider a directed multigraph whose edges types are the $k$-edges for those $k$ with $J_k > 0$, 
	$\overrightarrow{\mathcal{M}}^{[S]_{\tilde{H}}}$
	on $\mathcal{M}^{[S]_{\tilde{H}}}$. 
	For each $k$ with $J_k > 0$ in $\tilde{H}$,
	    there are $2^{k-1}$ currents 
	    $\bm{n}_{[s_k]} \in [S]_{\tilde{H}}$ pertaining to $k$ and each
	$\bm{n}_{[s_k]}$ describes the number of a particular directed $k$-edge and its reversal.
	To be exact, 
	${\bm{n}_{[s_k]}}_{(uv)}$ describes the number of those directed $k$-edges that have exactly the labelled edges 
	$i \in
	\{ j : (s_k)_j = + 1 
	\}
	$, 
	oriented from $u$ to $v$.
	This is well-defined up to fixing a representative of $[s_k]$.
	\end{remark}
	
	\subsection{Height model}
	Let $G^\dagger = (V^\dagger, E^\dagger)$ be the 
	dual graph of $G$. 
	Each tuple in $\Omega_0(G)$ naturally defines a height 
	function $h$ on the set of faces of $G$. 
	The difference between two neighbouring faces $u$ and 
	$u'$ is defined by 
	\begin{align*}
	h(u) - h(u')& = 
	(\bm{n}_{(vv')} - \bm{n}_{(v'v)})
	+ 
	2 ( {\bm{m_1}}_{(vv')} - {\bm{m_1}}_{(v'v)})
	+ 
	0 ( {\bm{m_2}}_{(vv')} - {\bm{m_2}}_{(v'v)}) 
	\\
	&= 
	(\bm{n}_{(vv')} - \bm{n}_{(v'v)})
	+ 
	2 ( {\bm{m_1}}_{(vv')} - {\bm{m_1}}_{(v'v)}),
	\end{align*}
	where the primal directed edge $(vv')$ crosses the 
	dual directed edge $(uu')$ from right to left.
	Well-definedness follows from the fact that 
	$\delta (\bm{n} + 
	2 \bm{m_1} + 
	0 \bm{m_2}) \equiv 0$.
	\par  
	Graphically one can think of the height difference between two neighbouring faces as the difference between the number of edges oriented from $v$ to $v'$ and the number of edges oriented from $v'$ to $v$ in 
	$\overrightarrow{\mathcal{M}}^{\bm{N}}$. 
	Here the two edges in the doubled edges count as two separate edges. 
	\par 
	The height function is defined up to a constant. 
	Let $\partial$ be the exterior vertex of $G^\dagger$, which in turn corresponds to the outer face of $G$.
	Then there is a one-to-one correspondence between
	\[ 
	\Omega^\dagger_0(G^\dagger) \coloneqq 
	\{ 
	h \in \mathbb{Z}^{V^\dagger}
	\ : \ 
	h(\partial) = 0 
	\}
	\]
	and $\Omega_0(G)$. 
	\par 
	Consequently $\mathbb{P}^{\tilde{H}}_{G,
	\beta}$ can also be seen as a probability measure on 
	the set of height functions $\Omega_0^\dagger$. 
	In terms of the height function, it is a Gibbs 
	measure given by
	\begin{equation}
	\label{eq:height def}
	\mathbb{P}^{\tilde{H}}_{G^\dagger, \beta}
	(h)
	\propto
	\exp \left(
	- \sum_{uu' \in E^\dagger} 
	\mathcal{V}^{\tilde{H}}_\beta ( h(u)- h(u'))
	\right).
	\end{equation}
	The symmetric potential
	$\mathcal{V}^{\tilde{H}}_\beta : 
	\mathbb{Z} \rightarrow \mathbb{R}$ is given by 
	\begin{align*}
	\mathcal{V}^{\tilde{H}}_\beta (k)
	&\propto	
	- \log 
	\left( 
	\sum_{n + 2 m_1 + 0 m_2 = k}
	I_n(\beta J_1)
	I_{m_1} (\beta J_2/2)
	I_{m_2} (\beta J_2/2)
	\right) 
	\\
	&\propto
	- \log 
	\left(
	\sum_{n + 2 m_1 = k}
	I_n(\beta J_1)
	I_{m_1} (\beta J_2/2)
	\right),
	\end{align*}
	where 
	\[ 
	I_k (\beta J) 
	\coloneqq 
	\sum_{i = 0}^{\infty} 
	\frac{ ( \beta J / 2)^{2 i + |k|} }{i! (i+|k|)!}
	, 
	\ k \in \mathbb{Z}, \ \beta J > 0,
	\]
	is the \textit{modified Bessel function}.
	Well-definedness follows from the identity 
	\[ 
	\sum_{k \in \mathbb{Z}} 
	I_k (\beta J) 
	= 
	e^{\beta J} 
	\]
	which in turn follows from plugging in $\theta = 0$ in the Fourier expansion
	\[ 
	e^{\beta J \cos (\theta)} 
	= 
	\sum_{k \in \mathbb{Z}} 
	I_k (\beta J) e^{i k \theta}.
	\]
	
	\begin{remark} 
	For general Hamiltonians 
	$\tilde{H} 
	\in 
	\tilde{\mathcal{H}}$
	one obtains the symmetric potential
	\[ 
	\mathcal{V}^{\tilde{H}}_\beta(l) 
	\propto 
	- 
	\log 
	\left( 
	\sum_{
	\left( 
	\sum_{[s] \in [S]_{\tilde{H}}}
	\divv
	( 
	[s] )
	n_{[s]} 
	\right)
	= 
	l
	}
	\prod_{
	[s] \in [S]_{\tilde{H}}}
	I_{n_{[s]}}
	\left( 
	\beta J_{\ell ([s])}/2^{\ell ([s]) - 1} 
	\right)
	\right) .
	\]
	\end{remark}
	\section{Delocalisation of the height function model}
	\label{sec:deloc of height}
	Recall that $\Gamma$ is a planar and locally finite graph that is invariant under the action of a lattice isomorphic to $\mathbb{Z}^2$.
	Let $\Gamma^\dagger = \left( V^\dagger, E^\dagger \right)$ be its planar dual and let $\mathcal{V}: \mathbb{Z} \rightarrow \mathbb{R}$ be a  \textit{convex on the integers} and symmetric potential.
	It is possible to extend the notion of a Gibbs measure of height functions on the \textit{infinite} lattice 
	$\Gamma^\dagger$ for a potential $\mathcal{V}$.
	We call a measure $\nu$ on height functions 
	$
	\varphi: V^\dagger \rightarrow 
	\mathbb{Z}
	$
	a \textit{Gibbs measure for $\mathcal{V}$} if it satisfies the \textit{Dobrushin--Lanford--Ruelle relation}, that is for all finite $\Lambda \subset V^\dagger$,
	\[ 	
	\nu_\Lambda (\cdot) 
	= 
	\int_{\mathbb{Z}^{V^\dagger}}
	\nu_\Lambda^\varphi
	(\cdot) 
	d \nu(\varphi),
	\]
	where
	$\nu_\Lambda$ is $\nu$ restricted to $\Lambda$ 
	and 
	$\nu_\Lambda^\varphi$ is the Gibbs measure defined as in \eqref{eq:height def} (with $\mathcal{V}$ instead of
	$\mathcal{V}^{\tilde{H}}_\beta$) 
	conditioned to be $\varphi$ on the boundary of $\Lambda$.
	Moreover, $\nu$ is said to be \textit{shift-invariant} if it is invariant under the action of the lattice on $\Gamma$.
	We refer the reader to \cite{lammers:2022height} for a more detailed account. 
	\par
	The model \textit{delocalises} 
	if there exist no shift-invariant Gibbs measures on $\Gamma$. 
	Otherwise the model \textit{localises}. 	
	\par  
	Fix a face $\bm{o}$ and define a sequence 
	$(B_N)_{N \geq 0}$ of finite subgraphs of $\Gamma$ such that $B_0$ contains $\bm{o}$ and $B_N \nearrow \Gamma$. 
	Consider the sequence of dual graphs 
	$(B_N^\dagger)_{N \geq 0}$. 
	For all $N \geq 0$, 
	$\bm{o} \in B^\dagger_N$,  
	$B^\dagger_N$ with its exterior vertex and adjacent edges removed is a subgraph of $\Gamma^\dagger$
	and the sequence of $B^\dagger_N$ with its exterior vertex and adjacent edges removed increases to $\Gamma^\dagger$.  
	\par
	In Theorem 2.7 of \cite{lammers:2024delocalisation} Lammers showed that if $\mathcal{V}$ (in addition to being convex and symmetric) is \textit{supergaussian},
	that is if its second derivative 
	\begin{align*} 
	\mathcal{V}^{(2)} : 
	\mathbb{Z} & \rightarrow \mathbb{R} 
	\\ 
	k & \mapsto 
	\mathcal{V}(k+1) - 2 \mathcal{V}(k) 
	+ \mathcal{V}(k-1)
	\end{align*}
	is non-increasing over the non-negative integers, localisation is equivalent to 
	\begin{equation} 
	\label{eq:loc var}
	\sup_{N \geq 0} 
	\mathbb{E}^\mathcal{V}_{B^\dagger_N, \beta} 
	\left[
	h(\bm{o})^2
	\right]
	 < \infty
	\end{equation}
	and 
	delocalisation is equivalent to 
	\begin{equation} 
	\label{eq:deloc var}
	\sup_{N \geq 0} 
	\mathbb{E}^\mathcal{V}_{B^\dagger_N, \beta} 
	\left[ 
	h(\bm{o})^2
	\right]
	 = \infty.
	\end{equation}
	\par 
	 In Theorem 5.3 of \cite{engelenburg:2025} van Engelenburg and Lis 
	 proved delocalisation of $\mathcal{V}^H$ for $H \in \mathcal{H}$ (and in particular for $\tilde{H} \in \tilde{\mathcal{H}}$) at finite inverse temperature $\beta$ and equivalently unbounded variance of the height function as in \eqref{eq:deloc var}. 
	
	\begin{theorem}[Theorem 5.3 \cite{engelenburg:2025}]
	\label{th:height deloc}
	Let 
	$\Gamma^\dagger$ and 
	$(B_N^\dagger)_{N \geq 0}$ as above and $H \in \mathcal{H}$ with at least one $J_k > 0$. 
	There exists 
	$\beta_0 \in (0, \infty)$ such that for all 
	$\beta \geq \beta_0$ , 
	\begin{equation} 
	\label{eq:deloc}
	\sup_{N \geq 0}
	\mathbb{E}_{B_N^\dagger, \beta}^H
	\left( h(\bm{o})^2 \right) 
	= 
	\infty .
	\end{equation}
	\begin{proof} 
	We only briefly summarize the proof as given in \cite{engelenburg:2025} for completeness and refer the reader to the original paper for further details.
	\par 	
	Take $H \in \mathcal{H}$ and pick the smallest $k \geq 1$ such that $J_k > 0$. 
	Let $H^\ast$ be the Hamiltonian with solely the coupling term corresponding to $k$ remaining, 
	that is 
	\[ 
	H^\ast = - \sum_{u \sim v} J_k \cos 
	\left( k ( \theta_u - \theta_v) \right). 
	\]
	A change of variables reduces the latter Hamiltonian to the standard XY model, where it was shown in \cite{lis:2023} that the potential of the dual height function model is convex, symmetric and supergaussian.
	Hence for the height function model on $\Gamma^\dagger$ with potential $\mathcal{V}^{H^\ast}$ at all positive inverse temperatures $\beta \in (0, \infty)$, the dichotomy in 
	\eqref{eq:loc var} and \eqref{eq:deloc var} holds.
	\par 
	By a generalized Ginibre inequality (Lemma 5.2 \cite{engelenburg:2025}) it follows that 
	$\mathbb{E}^{\mathcal{V}^{H^\ast}}_{\Gamma^\dagger, \beta}
	\left[ 
	h(\bm{o})^2 
	\right]$ is increasing in $\beta$, 
	which enables a lower bound of the type 
	\[
	\mathbb{E}^{{\mathcal{V}}^{H^\ast}}_{\Gamma^\dagger, \beta}
	\left[ 
	h(\bm{o})^2 
	\right]
	\leq 
	\mathbb{E}^{{\mathcal{V}}^{H}}_{\Gamma^\dagger, \beta}
	\left[ 
	h(\bm{o})^2 
	\right].
	\] 
	\par 
	Subsequently, it suffices to prove delocalisation of the height function model on $\Gamma^\dagger$ with potential $\mathcal{V}^{H^\ast}$ at large enough $\beta$. 
	To do this, one can make use of a sufficient criterion given by Lammers in Theorem 2.5 of \cite{lammers:2022height}, which says that if $T$ is a translation invariant graph of degree at most three, and if the potential $\mathcal{V}$ is, in addition to being convex and symmetric, \textit{excited}, 
	that is 
	\[ 
	\mathcal{V}(\pm 1) \leq \mathcal{V}(0) + \log 2,
	\]
	then the corresponding model delocalises.
	\par 
	In order to extend from a graph of degree at most three to the planar and biperiodic lattice $\Gamma$, one has to \textit{split}, \textit{glue} and \textit{reduce}
	(Lemma 5.7 - Lemma 5.9 \cite{engelenburg:2025}) edges.
	\end{proof}
	\end{theorem}
\section{Delocalisation rules out exponential decay} 
	\label{sec:loop repr}
	Let $\tilde{H} \in \tilde{\mathcal{H}}$ henceforth.
	If $\mathcal{V}^{\tilde{H}}_\beta$ delocalises, then
	\[ 
	\mathbb{E}^{\tilde{H}}_{B^\dagger_N, \beta}
	\left( 
	|h (\bm{o})| \right)
	\xrightarrow{N \rightarrow \infty}
	\infty.
	\]
	The next section aims to connect the expectation 
	$\mathbb{E}^{\tilde{H}}_{B^\dagger_N, \beta}
	\left( 
	|h (\bm{o})| 
	\right)$ of the height function model
	to 
	nematic order parameter $G_2$ of the generalized XY model with Hamiltonian $\tilde{H}$. 
	The idea is to decompose the expectation of a product of 
	spins under 
	$d \mu^{\tilde{H}}_{G, \beta}$ into a sum of the weights of loop and path 
	configurations.
	Then the height $h(\bm{o})$ is determined by the net winding of the loops around $\bm{o}$.
	Finally we have to compare the probability that two vertices are connected by a loop to the nematic order parameter $G_2$.
	\subsection{Loop representation}
	The starting point for the loop representation is the random current expansion of Section \ref{sec:current exp}.
	Let $\bm{N}$ induce $\mathcal{M}^{\bm{N}}$ and $\overrightarrow{\mathcal{M}}^{\bm{N}}$ on $\mathcal{M}^{\bm{N}}$.
	Note that there are precisely 
	\[ 
	\prod_{uv \in E} 
	\frac{|\bm{n}|_{uv}!}
	{\bm{n}_{(uv)}! \bm{n}_{(vu)}!}
	\frac{ (|\bm{m_1}| + |\bm{m_2}|)_{uv}!}
	{
	{\bm{m_1}}_{(uv)}!
	{\bm{m_1}}_{(vu)}!
	{\bm{m_2}}_{(uv)}!
	{\bm{m_2}}_{(vu)}!
	}
	\]
	$\overrightarrow{\mathcal{M}}^{\bm{N}}$ on 
	$\mathcal{M}^{\bm{N}}$.
	Recall the weight of $\bm{N}$ in \eqref{eq:weight current} and the number of single and doubled edges of a multigraph with single and doubled edges, $s(\cdot)$ and $d(\cdot)$. 
	With this notation we can rewrite:
	\[ 
	w(\bm{N}) 
	= 
	\sum_{
	\overrightarrow{\mathcal{M}}^{\bm{N}}
	\text{on
	$
	\mathcal{M}^{\bm{N}}
	$
	}
	}
	\prod_{uv \in E} 
	\frac{ 
	(\beta J_1/2)^{s(\mathcal{M}^{\bm{N}})_{uv}}}
	{ s(\mathcal{M}^{\bm{N}} )_{uv}!}
	\frac{ 
	(\beta J_2/4)^{d(\mathcal{M}^{\bm{N}})_{uv}}}
	{ d(\mathcal{M}^{\bm{N}} )_{uv}!}. 
	\]
	One can think of 
	\[ 
	\prod_{uv \in E} 
	\frac{ 
	(\beta J_1/2)^{s(\mathcal{M}^{\bm{N}})_{uv}}}
	{ s(\mathcal{M}^{\bm{N}} )_{uv}!}
	\frac{ 
	(\beta J_2/4)^{d(\mathcal{M}^{\bm{N}})_{uv}}}
	{ d(\mathcal{M}^{\bm{N}} )_{uv}!}
	\]
	as a weight assigned to each 
	$\overrightarrow{\mathcal{M}}^{\bm{N}}$, which is constant on $\mathcal{M}^{\bm{N}}$.
	\par 
	Specify a vertex set $S \subseteq V$. 
	Given a directed multigraph with doubled edges induced by $\bm{N}$, one can repeatedly glue an incoming and outgoing edge together at every vertex $v \in V \setminus S$ irrespective whether that edge is a single edge or belongs to a doubled edge.
	In the end one obtains a collection of directed loops not visiting $S$ or directed loops not visiting $S$ except possibly at their start or end vertex.
	At each vertex $v \in V \setminus S$, $\delta(\bm{n} + 2 \bm{m_2})_v$ outgoing edges remain unpaired if $\delta(\bm{n} + 2 \bm{m_2})_v > 0$ and $- \delta(\bm{n} + 2 \bm{m_2})_v$ incoming edges remain unpaired if $\delta(\bm{n} + 2 \bm{m_2})_v < 0$.
	If $\delta(\bm{n} + 2 \bm{m_2})_v = 0$, then all incoming and outgoing edges are paired. 
	\par 
	This motivates the following definition.
	\begin{definition}[Loop configurations on $	
	\mathcal{M}^{\bm{N}}$ outside 
	$S$]
	\label{def:loop conf}
	Let 
	$\bm{N} = 
	\left( \bm{n}, \bm{m_1}, \bm{m_2} 
	\right)$ 
	be a tuple of 
	currents. 
	Let 
	$\bm{N}$ 
	induce 
	$\mathcal{M}
	^{\bm{N}}$  and 
	$
	\overrightarrow{\mathcal{M}}^{\bm{N}}$ on 
	$\mathcal{M}^{\bm{N}}$ and let 
	$S \subseteq V$.
	A \textit{loop configuration on 
	$\mathcal{M}
	^{\bm{N}}$ outside $S$} is 
	\begin{itemize}
	\item defined on a directed multigraph with doubled 		
	edges 
	$\overrightarrow{\mathcal{M}}^{\bm{N}}$, and is 
	\item a collection of 
	     \begin{itemize}
	     \item unrooted directed loops on 
	     $
	     \overrightarrow{\mathcal{M}}^{\bm{N}}$ 
	     not visiting $S$, or
	     \item directed open paths on $	
	     \overrightarrow{\mathcal{M}}^{\bm{N}}$ 
	     not visiting $S$ except possibly at 
	     the start or end vertex such that   
	    \item every directed single edge is traversed exactly 
	    once by a loop or path 
	    \item every 'first' and 'second' edge of every 
	    directed doubled edge is traversed exactly once by a 
	    loop or path in their respective direction
	    \item at each vertex $v \in V \setminus S$, there are 
	    exactly $\delta ( \bm{n}+2 \bm{m_1} )_v$ outgoing 
	    paths if 
	    $\delta ( 
	    \bm{m} + 2 \bm{m_1} )_v > 0$ and $- 
	    \delta ( \bm{n}+2 \bm{m_1})_v$ 
	    incoming paths if $\delta ( \bm{n}+2 \bm{m_1} )_v < 0$.
    \end{itemize}
	\end{itemize}
	$\mathcal{L}^S_{\bm{N}}$ is the 
	set of all loop configurations on $\mathcal{M}
	^{\bm{N}}$ outside $S$.
	\end{definition}
	Next we want to assign a weight of loop configuration uniform among loop configurations defined on a specific
	$\overrightarrow{\mathcal{M}}^{\bm{N}}$. 
	\par 
	Let 
	$\deg_{\mathcal{M}^{\bm{N}}} 
	(u) 
	\coloneqq 
	\sum_{v \sim u} | \bm{n} |_{uv} 
	+ 2 | \bm{m_1} |_{uv}
	+ 2 | \bm{m_2} |_{uv}
	$ 
	be the degree of $\mathcal{M}^{\bm{N}}$ at vertex $u$.
	There are 
	\[ 
	\frac{ 
	(( \deg_{\mathcal{M}^{\bm{N}}}(v) + 	
	|\delta(\bm{n}+2 \bm{m_1})|_v)/
	2)!}
	{|\delta(\bm{n}+2 \bm{m_1})|	
	_v!}
	\] 
	ways to glue the incoming and outgoing edges together at vertex $v \in V \setminus S$ and therefore there are 
	\[ 
	\prod_{v \in V \setminus S}
	\frac{ 
	(( \deg_{\mathcal{M}^{\bm{N}}}(v) + 	
	|\delta(\bm{n}+2 \bm{m_1})|_v)/
	2)!}
	{|\delta(\bm{n}+2 \bm{m_1})|	
	_v!}
	\] 
	many loop configurations on $\overrightarrow{\mathcal{M}}^{\bm{N}}$.
	This suggests the following definition:
	
	\begin{definition}
	Let $\bm{N} = 
	\left( 
	\bm{n}, \bm{m_1}, \bm{m_2} 
	\right)$. 
	The weight for a loop configuration $\omega \in \mathcal{L}^S_{\bm{N}}$ is defined by
	\begin{multline}
	\label{eq: weight single loop}
	\lambda^S_{\bm{N}}(\omega) 
	\coloneqq
	\prod_{v \in V \setminus S} 
	\frac{|\delta(\bm{n}+2 \bm{m_1})|	
	_v!}
	{ 
	(( \deg_{\mathcal{M}^{\bm{N}}}(v) + 	
	|\delta(\bm{n}+2 \bm{m_1})|_v)/
	2)!}
	\\
	\prod_{uv \in E}
	\frac{ 
	(\beta J_1/ 2)^{s(\mathcal{M}^{\bm{N}})_{uv}}}
	{s(\mathcal{M}^{\bm{N}})_{uv}!}
	\frac{ (\beta J_2/4)^{d(\mathcal{M}
	^{\bm{N}})_{uv}}}
	{d(\mathcal{M}^{\bm{N}})_{uv}!}.
	\end{multline}
	\end{definition}
	
	With this weight,
	\begin{align*}
	w (\bm{N}) & = 
	\sum_{
	\overrightarrow{\mathcal{M}}^{\bm{N}} 
	\text{on 
	$\mathcal{M}^{\bm{N}}$
	}}
	\prod_{uv \in E}
	\frac{ (\beta J_1 / 2)^{s(\mathcal{M}^{\bm{N}})_{uv}}}
	{s(\mathcal{M}^{\bm{N}})_{uv}!}
	\frac{ (\beta J_2 / 4)^{d(\mathcal{M}
	^{\bm{N}})_{uv}}}
	{d(\mathcal{M}^{\bm{N}})_{uv}!} 
	\\
	& = 
	\sum_{
	\overrightarrow{\mathcal{M}}^{\bm{N}}
	\text{on 
	$\mathcal{M}^{\bm{N}}$}}
	\sum_{\omega 
	\
	\text{on 
	$\overrightarrow{\mathcal{M}}^{\bm{N}}$}}
	\lambda^S_{\bm{N}}(\omega) 
	\\
	& = 
	\sum_{\omega \in \mathcal{L}^S_{\bm{N}}}
	\lambda^S_{\bm{N}} (\omega).
	\end{align*}
	Let $\varphi, \psi : V \rightarrow \mathbb{Z}$. 
	Define 
	\[ 
	\Omega_\varphi \coloneqq \{ \bm{N} = (\bm{n},\bm{m_1}, 
	\bm{m_2}): \bm{n}, \bm{m_1}, \bm{m_2} \ 
	\text{currents and $\delta(\bm{n} + 2 \bm{m_2} ) = 
	\varphi$} \}.
	\]
	Then,
	\[
	\mathcal{L}^S_{\varphi} 
	\coloneqq 
	\bigcup_{\bm{N} \in \Omega_\varphi} 
	\mathcal{L}^S_{\bm{N}}.
	\]
	The union is disjoint, so each 
	$\omega \in \mathcal{L}^S_\varphi$ 
	is contained in exactly one 
	$\mathcal{L}^S_{\bm{N}}$.
	Hence the following weight on 
	$\mathcal{L}^S_\varphi$, 
	\[
	\lambda^S_\varphi (\omega) \coloneqq
	\lambda^S_{\bm{N}} (\omega) \ 
	\text{if} \ 
	\omega \in \mathcal{L}^S_{\bm{N}}
	\]
	is well-defined. 
	
	\begin{remark} 
	Loop configurations are defined analogously
	for general $\tilde{H} \in \tilde{\mathcal{H}}$.
	Denote the set of loop configurations on 
	$\mathcal{M}^{[S]_{\tilde{H}}}
	$
	by $\mathcal{L}^S_{[S]_{\tilde{H}}}$.
	The weight for $\omega \in \mathcal{L}^S_{[S]_{\tilde{H}}}
	$
	is then given by 
	\begin{multline*}
	\lambda^S_{[S]_{\tilde{H}}}
	(\omega) 
	= 
	\prod_{v \in V \setminus S} 
	\frac{ 
	| 
	\delta 
	\left( 
	\sum_{[s] \in [S]_{\tilde{H}}} 
	\divv ( [s] ) 
	\bm{n_{[s]}}
	\right)
	|_v 
	!}
	{
	\left( 
	\left( 
	\deg_{\mathcal{M}^{[S]_{\tilde{H}}}}(v)
	+
	| 
	\delta 
	\left( 
	\sum_{[s] \in [S]_{\tilde{H}}} 
	\divv ([s]) 
	\bm{n_{[s]}}
	\right) 
	|_v 
	\right) / 2
	\right) ! 
	}
	\\ 
	\prod_{uv \in E} 
	\prod_{k : J_k > 0 \
	\text{in} 
	\ 
	\tilde{H}}
	\frac{(\beta J_k/2^k)^{
	\sum_{ 
	[s_k] \in \{ \pm 1 \}^k / \sim} 
	| \bm{n_{[s_k]}}|_{uv} }
	}
	{
	\left(
	\sum_{ 
	[s_k] \in \{ \pm 1 \}^k / \sim} 
	| \bm{n_{[s_k]}}|_{uv}
	\right)!
	}.
	\end{multline*}
	\end{remark}
	
	With this notation and decomposition, we can follow
	\begin{align}
	Z^{\tilde{H}}_{G, \beta} & = 
	\sum_{\omega \in \mathcal{L}^S_0}
	\lambda^S_{0}(\omega)
	=
	\sum_{\omega \in \mathcal{L}_0^\emptyset} 	
	\lambda^\emptyset_0 (\omega) 
	\\
	\left( \sigma_a \bar{\sigma}_b 
	\right)^{\tilde{H}}_{G, \beta}
	& =
	\sum_{\omega \in \mathcal{L}^S_{\delta_a - \delta_b}}
	\lambda^S_{\delta_a - \delta_b} (\omega)
	\\
	\left( \sigma_a^2 \bar{\sigma}_b^2 
	\right)^{\tilde{H}}_{G, \beta}
	& =
	\sum_{\omega \in \mathcal{L}^S_{2(\delta_a - 
	\delta_b)}}
	\lambda^S_{2(\delta_a - \delta_b)} (\omega),
	\label{eq:two paths}
	\end{align}
	and so on. 
	Here, $0$ denotes the zero function on $V$. 
	\par
	$\lambda^\emptyset_0(\cdot)$ induces a probability 
	measure 
	$\textbf{P}^{\tilde{H}}_{G, \beta}$ on 
	$\mathcal{L}^\emptyset_0$.
	As long as $G$ is planar, its dual graph is well-defined and so is 
	$W_u(\omega)$ - the total net 
	winding of all the loops in loop configuration $
	\omega$ around a face $u \in V^\dagger$.
	Clearly,
	
	\begin{theorem}
	\label{th:winding}
	Let $\tilde{H} \in \tilde{\mathcal{H}}$.
	$(W_u)_{u \in V^\dagger}$ under 
	$\textbf{P}^{\tilde{H}}_{G, \beta}$ and 
	$(h(u))_{u \in V^\dagger}$ under 
	$\mathbb{P}^{\tilde{H}}_{G^\dagger,\beta}$ are identically 
	distributed.
	\end{theorem}	
	
	We are now in a position to compare the nematic order parameter $G_2(a-b)$ under $\mathbb{P}^{\tilde{H}}_{G, \beta}$ to probabilities of loop connections between $a$ and $b$ under $\bm{P}^{\tilde{H}}_{G, \beta}$.	
	The idea is to perform the operation of \textit{path switching}. 
	
	\begin{definition}[Path switching]	
	\label{def:path switching}	
	Let 
	$\varphi : V \rightarrow \mathbb{Z}$ such that 
	$\sum_{v \in V} \varphi_v = 0$.
	Suppose $\omega$ is a loop configuration in 
	$\mathcal{L}^S_\varphi$ with a path $\gamma$ going from $a$ to $b$, where $a,b \in S$.
	Let $r(\gamma)$ be the reversal of $\gamma$ and let 
	$R_\gamma(\omega)$ be the loop configuration identical to $\gamma$, except that the direction of the path $\gamma$ is reversed. 
	Then 
	\[ 
	\lambda^S_{\varphi} (\omega) 
	= 
	\lambda^S_{\varphi - 2(\delta_a - \delta_b)}
	(R_\gamma(\omega)) 
	\]
	because the underlying (undirected) multigraph with doubled edges and the number of outgoing edges minus incoming edges at each $v \in V \setminus S$ remains unchanged.
	\end{definition}
	
	\begin{remark} 
	\label{re:just tilde H}
	We wish to emphasize that the restriction to Hamiltonians $\tilde{\mathcal{H}}$ was necessary for the latter invariance under path switching to hold. 
	To be precise, loop configurations as in Definition \ref{def:loop conf} are defined on directed multigraphs with doubled edges, where (crucially) all four orientations of doubled edges can appear. 
	So switching a path $\gamma$ is firstly well-defined and secondly does not change the overall number of undirected doubled edges.
	\par  
	Had we rewritten the nematic terms 
	$\cos(2 (\theta_u - \theta_v)$ of $H_{\Delta}$ , $0 < \Delta < 1$,
	naively 
	as 
	\[ 
	\cos (2 (\theta_u - \theta_v) ) 
	= 
	\frac{
	(\sigma_u \bar{\sigma}_v)^2
	+ 
	(\bar{\sigma}_u \sigma_v)^2
	}{2},
	\]
	we would have obtained loop configurations on directed multigraphs with doubled edges, where only those directed doubled edges with both its first and second edge oriented in the same direction appear.
	Switching a path would then be ill-defined.
	The same holds for general $H \in \mathcal{H}$.
	\end{remark}

	\begin{lemma}
	\label{th:hilfslemma 1}
	Let $\tilde{H} \in \tilde{\mathcal{H}}$ with at least one $J_k > 0$ and $a$ and $b$ be two distinct vertices.
	Then,
	\begin{equation} 
	\label{eq:bound corr 1}
	\langle 
	\sigma_a^2 \bar{\sigma}_b^2 
	\rangle^{\tilde{H}}_{G, \beta}
	= 
	\bm{E}^{\tilde{H}}_{G,\beta}
	\left[
	\frac{m_{a,b}}{m_{a,b} + 1} 
	\chi_{\{m_{a,b} \geq 1\}} \right] 
	=
	\bm{E}^{\tilde{H}}_{G,\beta}
	\left[
	\frac{m_{a,b}}{m_{a,b} + 1} \right].
	\end{equation}
	Therefore,
	\begin{equation}
	\label{eq:bound corr}
	\frac{1}{2} 
	\bm{P}^{\tilde{H}}_{G,\beta}(m_{a,b} \geq 1)
	\leq
	\langle 
	\sigma_a^2 \bar{\sigma}_b^2 
	\rangle^{\tilde{H}}_{G, \beta}
	\leq
	\bm{P}^{\tilde{H}}_{G,\beta}(m_{a,b} \geq 1).
	\end{equation}
	\end{lemma} 	
	
	\begin{proof}
	The proof is analogous to Lemma 8 of \cite{lis:2023}. 	\par 
	Let $\varphi = 2(\delta_a - \delta_b)$.
	By \eqref{eq:two paths} we have
	\begin{equation}
	\label{eq:ausgeschrieben}
	    \langle \sigma_a^2 \bar{\sigma}_b^2 \rangle^{\tilde{H}}_{G, \beta}
	= 
	\frac{\sum_{\omega \in \mathcal{L}^{\{a,b\}}_\varphi}
	\lambda^{\{a,b\}}_\varphi
	(\omega)
	}
	{Z^{\tilde{H}}_{G, \beta}}
	\end{equation}
	
	We want to make use of the fact that $\omega \in \mathcal{L}^{\{a,b\}}_{\varphi}$ are those loop configurations with two more paths going from $a$ to $b$ than from $b$ to $a$. 
	If we switch one of the two additional paths from $a$ to $b$ we will end up with a loop configurations with as many paths from $a$ to $b$ as from $b$ to $a$ in the new loop configuration. 
	However there must be at least one path from $b$ to $a$. 
	By Definition \ref{def:loop conf} the weight is invariant under this switching.
	\par 
	To make everything precise we define $\mathcal{P}_{a,b}(\omega)$ to be the set of paths in $\omega$ going from $a$ to $b$. 
	Moreover, let
	\[
	m_{a,b}(\omega) = |\mathcal{P}_{a,b} (\omega)|.
	\]
	Each 
	$\omega \in \mathcal{L}^{\{a,b\}}_\varphi$ 
	induces 
	$m_{a,b}(\omega)$ pairs 
	$(\omega, \gamma)$, where 
	$\gamma \in \mathcal{P}_{a,b}(\omega)$. 
	\\
	Then the mapping 
	\begin{align*}
	R: 
	\{ (\omega, \gamma): 
	\omega \in \mathcal{L}^{\{a,b\}}_\varphi, 
	\gamma \in \mathcal{P}_{a,b}(\omega)
	%,m_{a,b}(\omega) \geq 2 
	\}
	& \rightarrow
	\{ (\omega, \gamma): 
	\omega \in \mathcal{L}^{\{a,b\}}_0, 
	\gamma \in \mathcal{P}_{b,a}(\omega),
	m_{a,b}(\omega) \geq 1
	\}
	\\
	(\omega, \gamma) 
	& \mapsto 
	(R_\gamma (\omega), r(\gamma))
	\end{align*}
	is well-defined.
	In each 
	$\omega \in \mathcal{L}^{\{a,b\}}_\varphi$ 
	there are two more paths from $a$ to $b$ than from $b$ to $a$, so
	\[
	m_{a,b}(\omega) 
	= m_{b,a}(\omega) + 2 
	\geq 2.
	\] 
	After reversing a path $\gamma \in \mathcal{P}_{a,b}(\omega)$ there are as many paths from $a$ to $b$ as from $b$ to $a$ but at least one path of both kinds, so 
	\[
	m_{a,b}(R_\gamma(\omega) )
	= m_{b,a}(R_\gamma(\omega) )
	\geq 1.
	\]
	The mapping $R$ is bijective. 
	Crucially,
	\begin{equation}
	\label{eq:weight unchanged}
	\lambda^{\{a,b\}}_\varphi (\omega)
	=
	\lambda^{\{a,b\}}_0(R_\gamma(\omega)) \
	\text{for} \
	\omega \in \mathcal{L}^{\{a,b\}}_\varphi, 
	\gamma \in \mathcal{P}_{a,b}(\omega),
	\end{equation}
	because the underlying undirected multigraph with doubled edges as well as the number of outgoing minus incoming edges outside $S= \{a,b\}$
	remains unchanged. 
	\\
	Now we can compute,
	\begin{align}
	\sum_{\omega \in \mathcal{L}^{\{a,b\}}_\varphi}
	\lambda^{\{a,b\}}_\varphi 
	(\omega)
	& = 
	\label{eq:hilf-1}
	\sum_{\omega \in \mathcal{L}^{\{a,b\}}_\varphi}
	\sum_{\gamma \in \mathcal{P}_{a,b}(\omega)}
	\frac{1}{m_{a,b}(\omega)}
	\lambda^{\{a,b\}}_\varphi 
	(\omega) 
	\\
	& = 
	\label{eq:hilf-2}
	\sum_{
	\substack{
	\omega \in \mathcal{L}^{\{a,b\}}_0: 
	\\
	m_{a,b}(\omega) \geq 1}}
	\sum_{\gamma \in \mathcal{P}_{b,a}(\omega)}
	\frac{1}{m_{a,b}(\omega) + 1}
	\lambda^{\{a,b\}}_0 (\omega) 
	\\
	& = 
	\label{eq:hilf-3}
	\sum_{
	\substack{
	\omega \in \mathcal{L}^{\{a,b\}}_0: 
	\\
	m_{a,b}(\omega) \geq 1}}
	\frac{m_{a,b}(\omega)}{m_{a,b}(\omega) + 1}
	\lambda^{\{a,b\}}_0 (\omega) 
	\\
	& =
	\label{eq:hilf-4}
	\sum_{
	\substack{
	\omega \in \mathcal{L}^{\{a,b\}}_0: 
	\\
	m_{a,b}(\omega) \geq 1}}
	\frac{m_{a,b}(\omega)}{m_{a,b}(\omega) + 1}
	\sum_{\omega' \in p^{-1}(\omega)}
	\lambda_0^\emptyset
	(\omega') 
	\\
	& = 
	\label{eq:hilf-5}
	\sum_{
	\substack{ 
	\omega' \in \mathcal{L}^\emptyset_0 
	\\
	m_{a,b}(\omega') \geq 1}}
	\frac{m_{a,b}(\omega')}{m_{a,b}(\omega') + 1}
	\lambda_0^\emptyset
	(\omega').
	\end{align}
	The second equality follows from the bijectivity of $R$, 
	$m_{a,b} (\omega) = m_{a,b}(R_\gamma(\omega)) + 1$
	and \eqref{eq:weight unchanged}, 
	the third one 
	because of 
	$m_{a,b}(\omega) = m_{b,a}(\omega)$ for $\omega \in \mathcal{L}^{\{a,b\}}_0$
	and 
	the fourth one because of the consistency property,  
	$p : \mathcal{L}^\emptyset_0 \rightarrow 
	\mathcal{L}^{\{a,b\}}_0$ cuts loop and path connections at $a$ and $b$.
	For the last equality note that $\omega' \in \mathcal{L}^\emptyset_0$, $m_{a,b}(\omega')$ refers to the number of pieces of loops going from $a$ to $b$ without visiting $a$ or $b$ in between. 
	If $\omega' \in p^{-1}(\omega)$, 
	$\omega \in \mathcal{L}^{\{a,b\}}_0$, then $m_{a,b}(\omega) = m_{a,b}(\omega')$. 
	\par 
	Inserting \eqref{eq:hilf-1} - \eqref{eq:hilf-5} into \eqref{eq:ausgeschrieben} proves \eqref{eq:bound corr 1}.
	\par 
	\eqref{eq:bound corr} follows by noting that 
	$\frac{1}{2} 
	\leq 
	\frac{m_{a,b}}{m_{a,b} + 1}
	\leq 1 
	$ 
	if $m_{a,b} \geq 1$.
	\end{proof}	
	
	As a by-product of the loop representation we record the following correlation inequality, which might turn out useful in the further study of generalized XY models.
	
	\begin{theorem}
	\label{th:corr ineq}
	Let $\tilde{H} \in \tilde{\mathcal{H}}$ and fix $a,b,c \in V$. 
	Let $B \subseteq V$ separate $a$ and $c$ in the sense that every path connecting $a$ to $c$ must pass trough a vertex of $b \in B$. 
	Then, 
	\begin{equation} 
	\langle 
	\cos (\theta_a + \theta_c - 2 \theta_b) 
	\rangle^{\tilde{H}}_{G, \beta}
	\leq 
	\langle 
	\cos (\theta_a - \theta_c) 
	\rangle^{\tilde{H}}_{G, \beta}
	\leq 
	\sum_{b \in B} 
	\langle 
	\cos (\theta_a + \theta_c - 2 \theta_b) 
	\rangle^{\tilde{H}}_{G, \beta}.
	\end{equation}
	\begin{proof}[Proof of Theorem \ref{th:corr ineq}]
	For the first inequality set $S = \{ b \}$.
	The first inequality is equivalent to 
	\begin{align*}
	( \sigma_a \sigma_c \bar{\sigma}^2_b )^{\tilde{H}}_{G, \beta}
	\leq 
	( \sigma_a \bar{\sigma}_c )^{\tilde{H}}_{G, \beta} 
	\\
	& \Leftrightarrow 
	\sum_{\omega \in \mathcal{L}^S_{\delta_a + \delta_c - 2 \delta_b}}
	\lambda^S_{\delta_a + \delta_c - 2 \delta_b} (\omega)
	\leq 
	\sum_{\omega \in \mathcal{L}^S_{\delta_a - \delta_c}}
	\lambda^S_{\delta_a - \delta_c} (\omega).
	\end{align*}
	If $\omega \in \mathcal{L}^S_{\delta_a + \delta_c - 2 \delta_b}$, then there exists a unique path $\gamma$ starting in $c$ and ending in $b$. 
	Switching this path maps 
	\[
	\omega \mapsto r_\gamma(\omega) \in \mathcal{E}_b
	\coloneqq 
	\{ \omega \in  \mathcal{L}^S_{\delta_a - \delta_c}: \ b \xrightarrow{\omega} c \}.
	\]
	One can recover $\omega$ by switching the reversed path $r(\gamma)$ in $r_\gamma(\omega)$ back, meaning $\omega = r_{r(\gamma)} (r_\gamma(\omega))$.
	This way we see that there is a one-to-one correspondence between
	\[
	\mathcal{L}^S_{\delta_a + \delta_c - 2 \delta_b} \ \text{and} \
	\mathcal{E}_b. 
	\]
	The weight remains unchanged under this switching because 
	$| (\delta_a - \delta_c)_v| = |(\delta_a + \delta_c - 2 \delta_b)_v|$ for all $v \in V \setminus \{b \}$. 
	\\
	Subsequently,
	\begin{align*}
	\sum_{\omega \in \mathcal{L}^S_{\delta_a + \delta_c - 2 \delta_b}}
	\lambda^S_{\delta_a + \delta_c - 2 \delta_b} 	(\omega)
	& =
	\sum_{ \omega \in \mathcal{E}_b}
	\lambda^S_{\delta_a - \delta_c} (\omega) 
	\\
	& \leq
	\sum_{ \omega \in \mathcal{L}^S_{\delta_a - \delta_c}} 
	\lambda^S_{\delta_a - \delta_c} (\omega) 
	\\
	& = 
	( \sigma_a \bar{\sigma}_c )^{\tilde{H}}_{G, \beta}.
	\end{align*}
	\par 
	For the second inequality set $S = B$.
	We can decompose
	\[
	\mathcal{L}^B_{\delta_a - \delta_c} 
	=
	\bigcup_{b \in B}
	\{
	\omega \in \mathcal{L}^B_{\delta_a - \delta_c} : 
	a \xrightarrow{\omega} b \},
	\]
	where the union is disjoint.
	Set $\mathcal{E}_b \coloneqq \{
	\omega \in \mathcal{L}^B_{\delta_a - \delta_c} : 
	a \xrightarrow{\omega} b \}$.
	Denote the unique path from $a$ to $b$ in $\omega \in \mathcal{E}_b$ by $\gamma$. 
	Switching this path maps 
	$\omega \mapsto r_\gamma(\omega) \in 
	\{ 
	\omega \in \mathcal{L}^B_{2 \delta_b - \delta_a - \delta_c} 
	:
	b \xrightarrow{\omega} a \} =
	\bar{\mathcal{E}}_b$. 
	By switching the unique path in $\bar{\mathcal{E}}_b$ from $b$ to $a$ back, we see that $\mathcal{E}_b$ and $\bar{\mathcal{E}}_b$ are in one-to-one correspondence. 
	\\
	The weights remain unchanged under this switching because 
	\[
	| (2 \delta_b - \delta_a - \delta_c)_v | 
	=
	| (\delta_a - \delta_c)_v |
	\]
	for all $v \in V \setminus B$.
	We conclude by
	\begin{align*}
	\left( \sigma_a \bar{\sigma}_c \right)^{\tilde{H}}_{G, \beta}
	 & = 
	\sum_{b \in b}
	\sum_{\omega \in \mathcal{E}_b}
	\lambda^B_{\delta_a - \delta_c}
	(\omega) 
	 = 
	\sum_{b \in B}
	\sum_{\omega \in \bar{\mathcal{E}}_b}
	\lambda^B_{2 \delta_b - \delta_a - \delta_c}
	(\omega) 
	\\
	 & \leq 
	\sum_{b \in B}
	\sum_{\omega \in \mathcal{L}^B_{2 \delta_b - \delta_a - \delta_c}}
	\lambda^B_{2 \delta_b - \delta_a - \delta_c}
	(\omega) 
	\\
	 & =
	\sum_{b \in B}
	\left(
	\sigma_a \sigma_c \bar{\sigma}^2_b 
	\right)^{\tilde{H}}_{G, \beta}.
	\end{align*}
	\end{proof}	
	\end{theorem}
	
	\subsection{Delocalisation and correlation}
	A \textit{cut at face $\bm{o}$} is a bi-infinite self-avoiding path in $
	\Gamma$ that goes through at least one edge incident 	
	to $\bm{o}$.
	As before, choosing an orientation and fixing a vertex incident 
	to $\bm{o}$ and visited by the cut, the cut splits into to 	infinite sets of vertices $L_{+}$ and $L_{-}$ - 				$L_{+}$ being the set of vertices that the cut visits 	before hitting the fixed vertex (including the fixed 			vertex) and $L_{-}$ being the set of vertices that 			the cut visits after hitting the fixed vertex.
	Crucially, each loop surrounding $\bm{o}$ must visit $L_+$ 
	as well as $L_-$.
	Let 
	\begin{align*}
	\chi^\epsilon_{\Gamma, \beta}
	(L) 
	& \coloneqq
	\sum_{a \in L_+ , b \in L_-}
	\left(
	\langle 
	\cos ( 2 (\theta_a - \theta_b) )
	\rangle^{\tilde{H}}_{\Gamma, 	
	\beta}
	\right)^{1 - \epsilon/2}
	\\ 
	& =
	\sum_{a \in L_+ , b \in L_-}
	G_2(a-b)^{1 - \epsilon/2}
	\
	,0 < \epsilon < 2,
	\end{align*}
	be the sum of the nematic order parameter along a cut $L$. 
	
	\begin{theorem}
	\label{cor:deloc impl no exp}
	Fix a Hamiltonian $\tilde{H} \in \tilde{\mathcal{H}}$.
	If the height function delocalises in the sense of 	
	\eqref{eq:deloc var}, then for all cuts $L$ and 
	$0 < \epsilon < 2$, 
	\begin{equation}
	\label{eq:statement cut}
	\chi^\epsilon_{\Gamma, \beta}(L) = \infty.
	\end{equation}
	In particular, $G_2$ does not decay exponentially in the distance.
	\end{theorem}
	
	In the proof of Theorem \ref{cor:deloc impl no exp} we make use of the following lemma, whose proof is postponed to the end of the section.
	
	\begin{lemma}
	\label{le:hilfslemma2}
	Let $\tilde{H} \in \tilde{\mathcal{H}}$ with at least one $J_k > 0$.
	For all finite graph $G=(V,E)$, $\beta > 0$ and $p > 
	1$, there exists $C_p < \infty$ (independent of 
	$G$) such that 
	\begin{equation}
	\label{eq:ineq lemma}
	\bm{E}^{\tilde{H}}_{G, \beta} 
	\left[
	m_{a,b} 
	\right]
	\leq
	C_p
	\deg_{G}(a) 
	\left(
	\bm{P}^{\tilde{H}}_{G,\beta}
	(m_{a,b} \geq 1)
	\right)^{1/p}
	.
	\end{equation}
	\end{lemma}
	
	\begin{proof}[Proof of Theorem \ref{cor:deloc impl no exp}]
	The proof is analogous to Proposition 9 of \cite{lis:2023}.
	\par 
	Suppose $G_2$
	decays exponentially in the distance, meaning there 
	exists $c, C>0$ such that
	\[
	\langle 
	\cos( 2 (\theta_u - \theta_v ))
	\rangle^{\tilde{H}}_{\Gamma, \beta}
	\leq 
	C 
	e^{-c |u-v| },
	\]
	then for any cut $L$ that is parallel to a straight line 
	\begin{equation}
	\label{eq:cuts infinity}
	\chi^\epsilon_{\Gamma, \beta} (L)
	\leq
	C
	\sum_{a \in L_+, b \in L_-}
	e^{-c (1-\epsilon/2) |a-b| }
	=
	C
	\sum_{n = 1}^\infty 
	n e^{- c (1-\epsilon/2) n}
	< \infty, 
	\end{equation}
	a contradiction.
	Vice versa, if we can show that $
	\chi^\epsilon_{\Gamma, \beta}(L) = \infty$ 
	for all cuts $L$, then $G_2$ does not decay 
	exponentially. 
	\par 
	With Lemma \ref{th:hilfslemma 1} and \ref{le:hilfslemma2} it is straightforward to 
	show 
	\eqref{eq:statement cut}:
	\\
	Recall that $W_{\bm{o}}(\omega)$ is the total net winding 
	around the distinguished face $\bm{o}$ of a loop 
	configuration $\omega$.
	Let $L$ be a cut at $\bm{o}$ in the infinite graph $
	\Gamma$ and $(B_N)_{N \geq 0}$, $(B^\dagger_N)_{N \geq 0}$ be defined as in the beginning of Section \ref{sec:deloc of height}. 
	Then,
	\begin{align}
	\mathbb{E}^{\tilde{H}}_{B^\dagger_N, \beta}
	\left[
	|h(\bm{o})|
	\right] 
	\label{eq:final equ 1}
	&=
	\bm{E}^{\tilde{H}}_{B_N, \beta}
	\left[
	|W_{\bm{o}}|
	\right] 
	\\
	&
	\leq
	\sum_{a \in L_+, b \in L_-}
	\bm{E}^{\tilde{H}}_{G_N, \beta}
	\left[
	m_{a,b}	
		\right]
	\label{eq:final equ 3}
	\\
	& 
	\leq 
	C 
	\sum_{a \in L_+, b \in L_-}
	\bm{P}^{\tilde{H}}_{B_N, \beta}
	\left(
	m_{a,b} \geq 1 
	\right)^{1/p}
	\label{eq:final equ 4} 
	\\
	& \leq
	\tilde{C} 
	\sum_{a \in L_+, b \in L_-}
	\langle 
	\sigma^2_a \bar{\sigma}^2_b 
	\rangle^{1/p}_{B´_N, \beta}
	\label{eq:final equ 5} 
	\\
	& \leq
	\tilde{C}
	\chi^\epsilon_{\Gamma, \beta}(L)
	\label{eq:final equ 6}
	\end{align}
	The first line is due to	Theorem \ref{th:winding} - 
	$W(\bm{o})$ under 
	$\bm{P}^{\tilde{H}}_{G, \beta}$ 
	is identically distributed to $h(\bm{o})$ under 
	$\mathbb{P}^{\tilde{H}}_{G, \beta}$. 
	For the second line note that each piece of loop that adds to the net winding 
	$W_{\bm{o}}(\omega)$ of a loop configuration $\omega$ must 
	intersect $L_+$ and $L_-$. 
	By a slight abuse of notation, we only sum over paths 
	intersecting the finite graph $B_N$. 
	The third line is Lemma \ref{le:hilfslemma2} with $1 / p = 1 - \epsilon / 2$, the fourth one is
	Lemma \ref{th:hilfslemma 1}, while 
	the last one is the second Ginibre inequality, Lemma \ref{le:ginibre}. 
	So if \eqref{eq:deloc} holds, then 
	$\chi^\epsilon_{\Gamma, \beta}(L) = 
	\infty$ for all cuts $L$.
	\end{proof}	
	\begin{proof}[Proof of Lemma \ref{le:hilfslemma2}]
	Fix $\beta > 0$, a finite graph $G=(V,E)$ and let $a,b \in V$. 
	Consider a loop configuration 
	$\omega \in \mathcal{L}^\emptyset_0$ on $G$.
	For simplicity let 
	$\tilde{H}$ be of the form \eqref{eq:H example}.
	A loop in $\omega$ traverses an edge $uv \in E$ either through a directed single edge or a directed doubled edge.
	Let 
	\begin{itemize}
	\item $\omega^s_{uv}$ 
	be the number of times that $\omega$ traverses $uv$ 
	through a single edge,
	\item $\omega^{d_1}_{uv}$ be the number of times that $\omega$ traverses $uv$ through a doubled edge, whose edges face the same direction and
	\item $\omega^{d_2}_{uv}$ 
	be the number of times that $\omega$ traverses $uv$ 
	through a doubled edge, whose edges face opposite 
	directions.
	\end{itemize}
	Note that 
	\[
	\sum_{c \sim a}
	\omega^s_{ac} +
	\omega^{d_1}_{ac} +
	\omega^{d_2}_{ac}
	\geq 
	m_{a,b}.
	\]
	Hence,
	\begin{multline*}
	\bm{E}^{\tilde{H}}_{G,\beta}
	\left[
	m_{a,b}
	\right]
	\leq
	\bm{E}^{\tilde{H}}_{G,\beta}
	\left[
	\sum_{c \sim a}
	\left(
	\omega^s_{ac} +
	\omega^{d_1}_{ac} +
	\omega^{d_2}_{ac}
	\right)
	\chi_{\{
	m_{a,b} \geq 1 \}}
	\right] 
	\\
	\leq 	
	\deg_G(a)
	\max_{c \sim a}
	\bm{E}^{\tilde{H}}_{G,\beta}
	\left[
	\left(
	\omega^s_{ac} +
	\omega^{d_1}_{ac} +
	\omega^{d_2}_{ac}
	\right)
	\chi_{\{
	m_{a,b} \geq 1 \}}
	\right] 
	\\
	\leq 	
	\deg_G(a)
	\max_{c \sim a}
	\left(
	\bm{E}^{\tilde{H}}_{G,\beta}
	\left[ 
	(\omega^s_{ac})^q
	\right]^{1/q}
	+
	\bm{E}^{\tilde{H}}_{G,\beta}
	\left[ 
	(\omega^{d_1}_{ac})^q
	\right]^{1/q}	
	+
	\bm{E}^{\tilde{H}}_{G,\beta}
	\left[ 
	(\omega^{d_2}_{ac})^q
	\right]^{1/q}
	\right) 
	\bm{P}^{\tilde{H}}_{G,\beta}(m_{a,b} \geq 1)^{1/p},
	\end{multline*}
	where the last inequality follows by H\"{o}lder with $1/p + 1/q = 1$. 
	\par 
	\eqref{eq:ineq lemma} follows if we are able to bound 
	$\bm{E}^{\tilde{H}}_{G,\beta}
	\left[ 
	(\omega^s_{ac})^q
	\right]$, respectively 
	$\bm{E}^{\tilde{H}}_{G,\beta}
	\left[ 
	(\omega^{d_1}_{ac})^q
	\right]
	$ and 
	$
	\bm{E}^{\tilde{H}}_{G,\beta}
	\left[ 
	(\omega^{d_2}_{ac})^q
	\right]$, 
	by a constant depending on $q$, $\beta$, $J_1$ and $J_2$ but not on the underlying graph $G$.
	\par 
	Note that for $e = uv$,
	\begin{align*}
	\bm{E}^{\tilde{H}}_{G,\beta}
	\left[ 
	(\omega^s_{uv})^q
	\right] 
	& =
	\mathbb{E}^{\tilde{H}}_{G,\beta} 
	\left[
	| \bm{n} |^q_{uv}
	\right] ,
	\\
	\bm{E}^{\tilde{H}}_{G,\beta}
	\left[ 
	(\omega^{d_{1/2}}_{uv})^q
	\right] 
	& =
	\mathbb{E}^{\tilde{H}}_{G,\beta} 
	\left[
	| 2 \bm{m}_{\bm{1}/\bm{2}} |^q_{uv}	
	\right].
	\end{align*}
	Define 
	$
	| \nabla h_{\bm{n}}  |_{uv}
	\coloneqq 
	| \bm{n}_{(uv)} - \bm{n}_{(vu)} |
	$
	for any current $\bm{n}$.
	Intuitively, this describes the difference between the number of single or doubled directed edges from $u$ to $v$ and its respective reversals. 
	\par
	Now 
	\begin{align*}
	\mathbb{E}^{\tilde{H}}_{G, \beta} 
	\left[ 
	|\bm{n}|^q_{uv} 
	\right] 
	& = 
	\sum_{i \geq 0} 
	i^q 
	\mathbb{P}^{\tilde{H}}_{G, \beta} 
	\left( 
	|\bm{n}|_{uv} = i
	\right) 
	\\
	& 
	= 
	\sum_{l \geq 0}
	\sum_{i \geq 0} 
	i^q 
	\mathbb{P}^{\tilde{H}}_{G, \beta} 
	\left( 
	|\bm{n}|_{uv} = i, 
	|\nabla h_{\bm{n}}|_{uv} = l 
	\right) 
	\\
	&= 
	\sum_{l \geq 0}
	\sum_{k \geq 0} 
	(2k+l)^q 
	\mathbb{P}^{\tilde{H}}_{G, \beta} 
	\left( 
	|\bm{n}|_{uv} = 2k+l, 
	|\nabla h_{\bm{n}}|_{uv} = l 
	\right) .
	\end{align*}
	Letting 
	\begin{align*}
	\Omega_l 
	& \coloneqq 
	\bigl\{ 
	\bm{N}:
	\delta 
	(\bm{n} + 2 \bm{m_1} ) =
	l(\delta_u - \delta_v) \bigr\} 
	\ \text{and} 
	\\
	\tilde{\Omega}_l 
	& \coloneqq 
	\bigl\{ 
	\bm{N}:
	\delta 
	(\bm{n} + 2 \bm{m_1} ) =
	l(\delta_u - \delta_v), 
	| \bm{n} |_{uv} = 0 \bigr\},
	\end{align*}
	we see that 
	\begin{align*}
	\mathbb{P}^{\tilde{H}}_{G, \beta} 
	\left( 
	|\bm{n}|_{uv} = 2k+l, 
	|\nabla h_{\bm{n}}|_{uv} = l 
	\right)
	& = 
	\frac{ 
	2
	\sum_{\bm{N} \in \tilde{\Omega}_l} 
	w(\bm{N}) 
	\frac{(\beta J_1/2)^{2k+l}}{k!(k+l)!} 
	}
	{\sum_{\bm{N} \in \Omega_0} 
	w(\bm{N})} 
	\leq 
	2 \frac{(\beta J_1/2)^{2k+l}}{k!(k+l)!} 
	\end{align*}
	Therefore 
	\begin{align*} 
	\mathbb{E}^{\tilde{H}}_{G, \beta}
	\left[ 
	|\bm{n}|^q_{uv} 
	\right] 
	& 
	\leq 
	2 
	\sum_{l \geq 0} 
	\sum_{k \geq 0} 
	\frac{(\beta J_1/2)^{2k+l}} 
	{k! (k+l)!}
	\\ 
	& = 
	\sum_{k \geq 0} 
	\sum_{m \geq k} 
	\frac{(\beta J_1/2)^{l + m} 
	}{k! m!}
	\\ 
	&
	\leq 
	 \sum_{k \geq 0} 
	\sum_{m \geq 0} 
	\frac{(\beta J_1/2)^{l + m} 
	}{k! m!}
	\\ 
	& 
	= 
	\sum_{n \geq 0} 
	\frac{(\beta J_1)^n}{n!},
	\end{align*}
	where we performed a change of variables
	and used the identity $\sum_{k = 0}^n {n \choose k} = 2^n$.
	The latter series is recognised to be the \textit{$q$-th Touchard polynomial} multiplied by $e^{\beta J_1}$. 
	The $q$-th Touchard polynomial is formally defined as 
	$T_q(\lambda) = e^{-\lambda} \sum_{n \geq 0} 
	\frac{\lambda^n 
	n^q}{n!}
	$
	and evaluates the $q$-th moment of a random variable $X \sim \Poi(\lambda)$.
	In particular this quantity is finite and depends only on $\beta$ and $J_1$.
	\par 
	Analogous bounds hold for $\bm{E}^{\tilde{H}}_{G, \beta} 
	\left[ 
	(\omega^{d_1/d_2}_{uv})^q 
	\right] 
	$. 
	\par 
	This extends in straightforward fashion to Hamiltonians $\tilde{H} \in \tilde{\mathcal{H}}$.
	\end{proof} 
\section{Proofs of the main results}
	\label{sec:main results} 
	\begin{proof}[Proof of Theorem \ref{th:main theorem 1}]
	The first part is Theorem \ref{cor:deloc impl no exp}.
	The second part is Theorem \ref{th:height deloc}. 
	\end{proof}
	\begin{proof}[Proof of Theorem \ref{th:main theorem 2}]	
	Recall that $\beta^{XY}_c$ is the critical inverse temperature of the standard XY model.
	Let 
	\begin{align*} 
	H_{\Delta,1} 
	& \coloneqq
	- 
	\sum_{u \sim v} 
	\Delta 
	\cos (\theta_u - \theta_v) 
	\\ 
	H_{\Delta,2} 
	& \coloneqq
	- 
	\sum_{u \sim v} 
	(1-\Delta) 
	\cos \left( 2 (\theta_u - \theta_v) \right).
	\end{align*}
	For $\Delta \in (0, 1)$ and $i = 1,2$,
	\begin{equation}
	\label{eq:G_i}
	G_i(v) \geq 
	\langle 
	\cos 
	\left( 
	i 
	(
	\theta_o - \theta_v 
	)
	\right) 
	\rangle_{\Gamma, \beta}^{H_{\Delta,i}}
	\end{equation}
	by Lemma \ref{le:ginibre}.
	Now for the existence of the nematic phase at 
	$\Delta \approx 0$ take $\beta > \beta^{XY}_c$ and choose
	$\tilde{\Delta}$ small enough such that $\beta (1 - \tilde{\Delta}) > \beta^{XY}_c$. 
	This yields a power-law lower bound of the form
	\[
	\langle 
	\cos 
	\left( 
	2
	(
	\theta_o - \theta_v 
	)
	\right) 
	\rangle_{\Gamma, \beta}^{H_{\tilde{\Delta},2}}
	\geq 
	\frac{1}{8 |v|}. 
	\]
	Subsequently, 
	$
	\langle 
	\cos 
	\left( 
	2
	(
	\theta_o - \theta_v 
	)
	\right) 
	\rangle_{\Gamma, \beta}^{H_{\Delta',2}}$
	possesses the same power-law lower bound for all 
	$
	\Delta' \leq \tilde{\Delta}$.
	\par
	By Theorem \ref{th:comp AT} we can upper bound
	\[
	\langle 
	\cos 
	\left( 
	\theta_o - \theta_v 
	\right) 
	\rangle_{\Gamma, \beta}^{H_{\Delta',1}}
	\leq 
	\langle 
	s_o s_v
	\rangle^{AT}_{\Gamma, \beta} ,
	\]
	where one the right-hand side the expectation is of the AT model with $J = \frac{\Delta'}{2}$ and $U = 1 - \Delta'$. 
	\par 
	By \eqref{eq:crit temp},  
	$\beta_c	^{(1)} \xrightarrow{\Delta' \rightarrow 0} \infty$, where 
	$\beta_c^{(1)}$ is the critical inverse temperature for the two-point function of the Ashkin--Teller model with $J = \frac{\Delta'}{2}$ and $U = 1 - \Delta'$. 
	Subsequently choosing $\Delta_0(\beta)$ small enough, 
	it holds that for all $0 \leq \Delta' \leq \Delta_0(\beta)$,
	\[ 
	\langle 
	s_o s_v
	\rangle^{AT}_{\Gamma, \beta},
	\]
	decays exponentially,
	where the expectation is still of the AT model with 
	$J = \frac{\Delta'}{2}$ and $U = 1 - \Delta'$.
	This finishes the proof.
	\end{proof}

\bibliographystyle{plain}
\bibliography{biblio}

\end{document}